\documentclass[reqno]{amsart}

\usepackage {amssymb}
\usepackage{graphicx}

\begin{document}

\newtheorem{lemma}{Lemma}
\newtheorem{theorem}{Theorem}
\newtheorem{proposition}{Proposition}
\newtheorem{corollary}{Corollary}
\newtheorem{hypothesis}{Hypothesis}
\newtheorem{remark}{Remark}
\newtheorem{definition}{Definition}

\newcommand{\bee}{\begin{eqnarray}}
\newcommand{\eee}{\end{eqnarray}}
\newcommand{\be}{\begin{eqnarray*}}
\newcommand{\ee}{\end{eqnarray*}}
\newcommand{\R}{{\mathbb R}}
\newcommand{\N}{{\mathbb N}}
\newcommand{\Z}{{\mathbb Z}}
\newcommand{\C}{{\mathbb C}}

\newcommand{\Rp}{{\Bbb R}}
\newcommand{\Np}{{\Bbb N}}
\newcommand{\Zp}{{\Bbb Z}}
\newcommand{\Cp}{{\Bbb C}}
\newcommand{\D}{\mbox {\sc D}}
\newcommand{\I}{\mbox {\sc 1}}
\newcommand{\0}{\mbox {\sc 0}}
\newcommand{\K}{{\it K}}
\newcommand{\II}{{\it I}}
\newcommand{\W}{{\it W}}
\newcommand{\F}{{\it F}}
\newcommand{\U}{{\it U}}
\newcommand{\V}{{\it V}}
\newcommand{\f}{\mbox {\sf f}}
\newcommand{\g}{\mbox {\sf g}}
\newcommand{\h}{\mbox {\sf h}}
\newcommand{\s}{{\it S}}

\newcommand{\ind}{\hskip 0.5cm}
\newcommand{\case}[2]{\textstyle{\frac{#1}{#2}}}

\newcommand{\E}{{\mathcal E}}
\newcommand{\asy}{\tilde {\mathcal O}}

\title []{Nonlinear Schr\"odinger equations with a multiple-well potential and a Stark-type perturbation}

\author {Andrea SACCHETTI}

\address {Department of Physics, Computer Sciences and Mathematics, University of Modena e Reggio Emilia, Modena, Italy\\Via G. Campi 213/A, Modena - 41125 - Italy}

\email {andrea.sacchetti@unimore.it}

\date {\today}

\begin {abstract} {A Bose-Einstein condensate (BEC) confined in a one-dimensional lattice under the effect of an external homogeneous field is described by the Gross-Pitaevskii 
equation. \ Here we prove that such an equation can be reduced, in the semiclassical limit and in the case of a lattice with a finite number of wells, 
to a finite-dimensional discrete nonlinear Schr\"odinger equation. \ Then, by means of numerical experiments we show that the BEC's center of mass exhibits an oscillating 
behavior with modulated amplitude; in particular, we show that the oscillating period actually depends on the shape of the initial wavefunction of the condensate as well as on the strength of 
the nonlinear term. \ This fact opens a question concerning the validity of a method proposed for the determination of the gravitational constant by means of 
the measurement of the oscillating period.}

\bigskip

%
%
%

\end{abstract}

\maketitle


\normalsize

%
%
%

\section {Introduction}

Laser-cooled atoms have drawn a lot of attention as for potential applications to interferometry and high-precision 
measurements, from the determination of gravitational constants to geophysical 
applications \cite {FFMK,LBCPT,MFFSK,RSCPT}, see also \cite {CladeReview,TinoReview} for a recent review. \ The idea of using ultracold atoms moving 
in an accelerated optical lattice \cite {Bloch1,Bloch2,RSN,SPSSKP,Shin} has opened the field to multiple applications. \ In particular, 
by means of the method proposed by Clad\'e {\it et al} \cite {CGSNJB}, a value for the constant $g$ has been measured using 
ultracold strontium atoms confined in a vertical optical lattice \cite {FPST}; such a result has been improved 
 by using a larger number of atoms and reducing the initial temperature of the sample \cite  {PWTAPT}. \ Determination of $g$ 
has been obtained by measuring the period $T$ of the Bloch oscillations of the atoms in the vertical optical 
lattice; recalling that 
\bee
T = \frac {2 \pi \hbar}{m g b} \, , \label {Eq0}
\eee
where $m$ is the mass of the Strontium atom, $\hbar$ is the Planck constant and $b$ 
is the lattice period, then a precise value of the constant $g$ has been obtained by means of the experimental measurements of the oscillating period. \ Since Bloch 
oscillations with period (\ref {Eq0}) have been predicted by the Bloch Theorem \cite {Callaway} only for a one-body particle in a periodic field and under the effect 
of a Stark potential then it has been chosen, in the experiments above, a particular Strontium's isotope ${}^{88}Sr$; in fact, the 
scattering length $a_s$ of atoms ${}^{88}Sr$ is very small and thus it has been assumed by \cite {FPST,PWTAPT} that the effects of the atomic 
binary interactions are negligible. \ The obtained value for the constant $g$ was consistent with the one obtained by classical gravimeters; but it was 
affected by a relative uncertainty of order $6 \times 10^{-6}$ because of a larger scattering in repeated measurements, mainly due to 
the initial position instability of the trap. \ Such a technique is also proposed to 
measure surface forces \cite {SAFIPST}, too.

The critical point of this experimental procedure concerns the validity of the Bloch Theorem and the estimate of the effect of the atomic binary interactions on the 
oscillating period of the BEC. \ In order 
to discuss this point here we are inspired by a realistic model of a one-dimensional cloud of cold atoms in a periodical optical 
lattice under the effect of the gravitational force. \ The periodic potential has the shape
\bee
V_{per} (x) = V_0 \sin^2 (k_L x) \label {Eq1}
\eee
where $b=\frac 12 \lambda_L$ is the period, and $\lambda_L =\frac {2\pi }{k_L}$. \ The one-dimensional BEC is governed by the one-dimensional time-dependent 
Gross-Pitaevskii equation with a periodic potential and a Stark potential
\bee
i \hbar {\partial_t \psi} = H_B \psi +  f x \psi + \gamma |\psi |^{2 } \psi \, , \ f = m g \, ,  \label {Eq2}
\eee
where the wavefunction $\psi (\cdot , t ) \in L^2 (R ,dx)$ is normalized to one: 
\be
\| \psi (\cdot , t ) \|_{L^2} = \| \psi_0 (\cdot )\|_{L^2} =1 \, , 
\ee
and where 
\be
H_B = - \frac {\hbar^2}{2m} \partial^2_{xx} + V_{per} (x) 
\ee
is the Bloch operator with periodic potential $V_{per}(x)$. \ By $\gamma $ we denote  the effective one-dimensional nonlinearity strength. 

It is a well known fact (see \S 6.1 by \cite {Callaway}) that when the wavefunction $\psi$ is prepared on the first band of the Bloch operator and if the nonlinear term is absent, i.e. $\gamma =0$, then the 
dominant term of the wavefunction $\psi$ exhibits a periodic behavior with Bloch period $T$ within an interval with amplitude 
$\frac {B_1}{|f|}$, where $B_1$ is the width of the first band and where $f \in \R$ is the strength of the external homogeneous field 
(in the case of $f = m g$ then $f$ takes only positive values, obviously). \ Therefore, for times of the order of the Bloch period $T$ we may assume that the 
motion of the BEC occurs in a finite interval. \ Hence, 
we can restrict ourselves to the analysis of equation (\ref {Eq2}) in a suitable finite interval and then we may assume to consider a 
multiple-well potential $V_N (x)$ (with a fixed number $N$ of wells) and that the Stark potential $x$ is replaced by a Stark-type 
potential $W_N (x)$ due to an homogeneous external field which acts only in a bounded region containing the $N$ wells (see Fig. \ref {Figura1}). \ That is, 
instead of (\ref {Eq2})  we consider, as a model for a BEC in an optical lattice under an external homogeneous field,  
the time-dependent non-linear Schr\"odinger equation (NLS)
\bee
\left \{ 
\begin {array}{l}
i \epsilon \partial_t \psi = H_N \psi  + f W_N(x) \psi + \gamma |\psi |^2 \psi \, , \ H_N = - \epsilon^2 \partial^2_{xx} + V_N   \\ 
\psi (x,0) = \psi_0 (x) 
\end {array}
\right. \label {Eq3}
\eee
where $\epsilon >0$ plays the role of the semiclassical parameter (we prefer to denote here the small semiclassical parameter by $\epsilon$ 
instead of the usual notation $\hbar$ because in a subsequent section we'll discuss a real physical model where $\hbar$ will assume 
its \emph {fixed} physical value; with such a 
notation it turns out that the Bloch period is given by $T= \frac {2\pi \epsilon}{|f|b}$). \ We assume that 
the $N$ wells have all the same shape and we denote by $b>0$ the distance between the adjacent absolute minima points. 
\begin{figure} [h]
\begin{center}
\includegraphics[height=6cm,width=8cm]{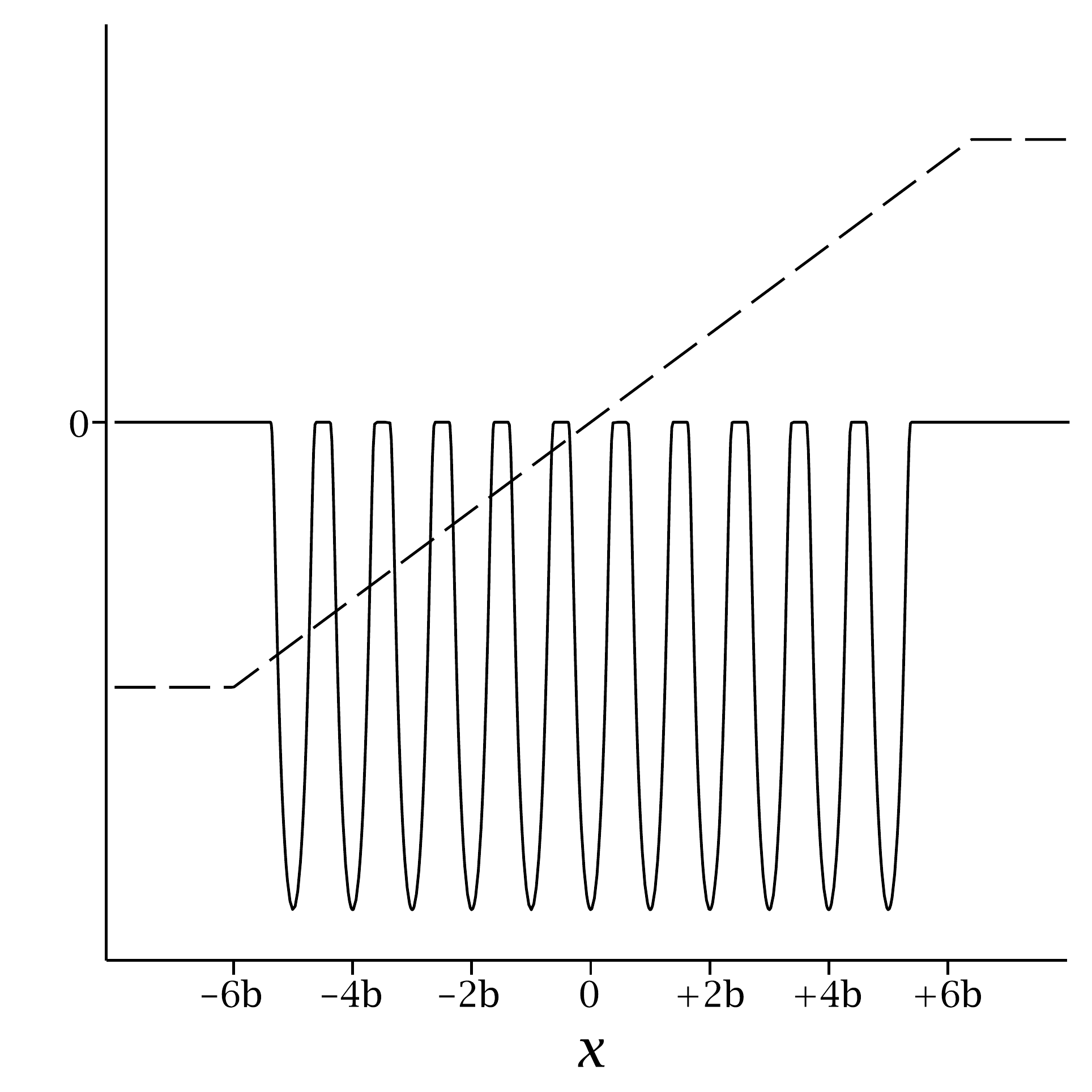}
\caption{\label {Figura1} Plot of the multiple-wells potential $V_N$ (full line) and of the Stark-type potential $W_N$ (broken line), where $N=11$. \ By $b>0$ 
we denote the distance between the adjacent absolute minima points.}
\end{center}
\end{figure}

The study of the dynamics of the wavefunction $\psi$, solution of (\ref {Eq3}), is then achieved by means of a discrete nonlinear Schr\"odinger 
equation (DNLS). \ The idea is basically simple and it consists in assuming that the wavefunction $\psi$ may be written as a superposition 
of vectors $u_\ell (x)$ localized on the $\ell -$th cell of the lattice; that is 
\be
\psi (x,t) \sim \sum_{\ell = 1}^N c_\ell (t) u_\ell (x) \, . 
\ee
Such an approach has been successfully used in the cases of semiclassical NLS with multiple-well potentials \cite {S1} or with periodic 
potentials 
(see \cite {FS,P1,P2}), without the external field with potential $W_N$. \ Eventually, $u_\ell (x)$ may coincide with the Wannier function $u_\ell^W (x)$ associated to the first band of the Bloch 
operator $H_B$ or with the semiclassical single well ground state eigenfunction $u_\ell^{sc} (x)$. \ By means of such an approach the unknown functions $c_\ell (t)$ turn out to be the solutions of a system of time-dependent equations  
which dominant terms are given by (here we denote $\dot {} = \frac {d}{dt}$)
\bee
i \epsilon \dot c_\ell = - \lambda_D c_\ell - \beta \left ( c_{\ell +1} + c_{\ell -1} \right ) + \gamma \| u_0 \|^{4}_{L^4} |c_\ell|^{2} c_\ell + 
f b \ell c_\ell \, , \ \ell =1,\ldots , N \label {Eq4}
\eee
where $\lambda_D$ is the ground state of a single cell potential and where $\beta $ is the hopping 
matrix element between neighboring sites. \ In fact, the parameter $\beta$ is expected to be such that $4 \beta$ is equal to the 
amplitude $B_1$ of the first band \cite {S2}. \ In (\ref {Eq4}) we'll fix $c_0 \equiv c_{N+1} \equiv 0$. \ Equation (\ref {Eq4}) represents a discrete nonlinear Schr\"odinger 
equation (DNLS).

Our approach is both semiclassical and perturbative. \ It is semiclassical in the sense that it holds true in the semiclassical regime of $\epsilon$ small enough; and 
it is perturbative in the sense that the external field $f$ and the nonlinearity power strength $\gamma$ must be small when $\epsilon$ goes to zero 
(see Hyp. \ref {Ipo3} for details). \ Under these conditions we prove the validity of the $N$-mode approximation (\ref {Eq4}) with a rigorous estimate of the remainder term 
\emph {for times of the order of the Bloch period}. \ Then, we numerically solve the $N$-mode approximation (\ref {Eq4}), and we compute the oscillating period taking into account the nonlinear 
interaction. \ In fact, the behavior of the wavefunction is not simply periodic in time; it turns out that the center of mass $\langle x \rangle^t = 
\langle \psi , x \psi \rangle $ shows an oscillating motion with modulated amplitude. \ The oscillating period turns out to be depending on the 
nonlinearity parameter strength $\gamma$ and we see that it also depends on the distribution of the initial wavefunction $\psi_0$. \ In particular, when $\psi_0$ is a 
symmetric wavefunction then the oscillating period is almost constant for small $\gamma$ and it practically coincides with the Bloch period $T$; on the other hand 
when $\psi_0$ is an asymmetrical function the oscillating period actually depends on $\gamma$. \ This fact is in contradiction with the Bloch Theorem 
(which holds true when $\gamma =0$), which implies that the Bloch period $T$ does not depend on the shape of the initial wavefunction, and it  
may explain the relatively large uncertainty observed by \cite {PWTAPT} in their experiments, as discussed in the Conclusions.

The paper is organized as follows. \ In Section 2 we derive the DNLS (\ref {Eq4}) from the NLS (\ref {Eq3}) in the semiclassical limit $\epsilon \to 0$ for times of the 
order of the Bloch period $T$ with a rigorous estimate of the remainder term. \ In particular: in \S 2.1 we introduce the assumptions and we recall some preparatory 
results; in \S 2.2 we derive the DNSL by making use of some ideas previously given by \cite {S1} and adapted to the case of multiple-well potential with an external Stark-type 
perturbation. \ In Section 3 we consider a realistic experiment and we compute the wavefunction dynamics by making use of the DNLS. \ In particular: in \S 3.1 
we discuss the validity of the $N$-mode approximation for different values of the parameters; in \S 3.2 we numerically compute the wavefunction for times of the order 
of the 
Bloch period. \ In Appendix we write the Wannier functions in terms of the Mathieu functions.

\subsection* {Notation} Let $g$ be a quantity depending on the semiclassical parameter $\epsilon$. \ In the following 
\be
g = \asy \left ( e^{-S_0/\epsilon } \right ) 
\ee
means that for any $\epsilon^\star >0$ and any $\rho \in (0, S_0 )$ there exists $C:=C_{\rho , \epsilon^\star}$ such that 
\be
|g| \le C e^{-(S_0 - \rho ) /\epsilon } \, , \ \forall \epsilon \in (0,\epsilon^\star) \, . 
\ee

Hereafter, by $C$ we denote a generic positive constant independent of $\epsilon$.

Let $N\in \N$, then by $\N_N := \{ 1,2,\ldots , N \}$ we denote the set of first $N$ positive integer numbers.

By $\| \cdot \|_{L^p}$ we denote the norm of the Banach space $L^p (\R )$, by $\langle \cdot , \cdot \rangle $ we denote the scalar product of the Hilbert space $L^2 (\R )$.

\section {Derivation of the DNLS (\ref {Eq4})}

\subsection {Assumptions and preliminary results}

We consider the time-dependent non-linear Schr\"odinger equation (\ref {Eq3}) where $V_N$ is a multiple-well potential and $W_N (x)$ is a bounded Stark-type potential. \ In 
particular we assume that

\begin {hypothesis} \label {Ipo1} {Let $v(x) \in C_0^\infty (\R)$ be an even (i.e. $v(-x)=v(x)$) smooth function with compact support with a non degenerate minimum value at $x=0$:
\be
v(x) > v_{min} = v(0) , \ \forall x \in \R \, ,\ x \not= 0 .
\ee
The multiple-well potential is defined as 
\be
V_N(x) = \sum_{\ell =1}^N v (x-x_\ell ) 
\ee
for some fixed $N>1$, where $x_\ell = \left ( \ell - \frac {N+1}{2} \right ) b $ and where $b>0$ is such 
that $\mbox {supp} \ v \subset \left ( - \frac b2 ,+ \frac b2 \right )$.}
\end {hypothesis}

Hence, by construction the potential $V_N(x)$ has exactly $N$ wells with not degenerate minima at $x=x_\ell $, $\ell \in \N_N$.

\begin {remark} \label {Remark1}
We assume that $v(x)$ is an even function just for argument's sake. \ As discussed in Remark \ref {Remark6} this assumption may be removed. \ Furthermore, we assume 
that $v$ is a smooth function as usual; in fact, a lessere regularity (e.g. $C^2$) would be enough.
\end {remark}

\begin {hypothesis} \label {Ipo2} {Let $W_N(x) \in C (\R)$ be the monotone not decreasing function defined as 
\be
W_N(x) = 
\left \{
\begin {array}{ll}
-L& \mbox { if } x < -L \\ 
 x & \mbox { if } x \in [-L , L ] \\
L & \mbox { if } x > L
\end {array}
\right.
\ee
for some $L > \frac {N+1}{2} b$.}
\end {hypothesis}

That is the Stark-type potential $W_N$ is linear in the region containing the wells and it is a constant function outside this region (see Fig. \ref {Figura1}). \ In 
the ``limit'' where $N$ goes to infinity the potential $V_N$ becomes a periodic potential with period $b$ and the external potential $W_N$ 
becomes the Stark potential $x$.

\begin {remark} \label {Remark2}
We restrict ourselves to a multiple-well potential $V_N$ with a finite number of wells only for sake of simplicity; one could consider the case of a periodic potential by 
making use of the tools developed by \cite {FS}. \ On the other side, the assumption on $W_N$ is not merely for the sake of simplicity; actually, the Stark-type potential 
$W_N$ is a bounded operator while the Stark potential $x$ is not a bounded operator and this fact is a source of several 
technical problems. \ In fact, in real experiments the BEC are trapped in a finite spatial region. 
\end {remark}

\begin {hypothesis} \label {Ipo3} {We assume to be in the {\bf semiclassical limit}, that is we look for the solution of (\ref {Eq3}) in the limit of $\epsilon $ that 
goes to zero. \ We assume also that the other two parameters 
$\gamma$ and $f$ are small for $\epsilon$ small. \ That is we assume that there exists $\epsilon^\star >0$ such that 
\be
C e^{- (S_0 - \rho )/\epsilon } \le |f| \le C \epsilon^s \, , \ \forall \epsilon \in (0, \epsilon^\star ) \, , 
\ee
for some $s > 2$, $C>0$ and $\rho \in (0, S_0)$ independent of $\epsilon$; furthermore, we assume also that 
\bee
\frac {|\gamma | \epsilon^{-1/2}}{|f|} \le C
\eee
for some positive constant $C$ and for any $\epsilon \in (0,\epsilon^\star )$.}
\end {hypothesis}

The self-adjoint extension of the linear Schr\"odinger operator formally defined on $L^2 (\R )$ as 
\be
H_N = - \epsilon^2 \partial^2_{xx} + V_N 
\ee
has an almost degenerate ground state with dimension $N$. \ More precisely, let $\lambda_\ell$, $\ell \in \N_N$, be the lowest eigenvalues of $H_N$ with 
associated normalized eigenvectors $v_\ell$. \ In particular we have that (see Lemma 2 \cite {S2})
\be
\lambda_\ell = \lambda_D - 2 \beta \cos \left ( \ell \frac {\pi}{N+1} \right ) + O(\epsilon^\infty ) e^{-S_0/\epsilon } \, , \ 
\ell \in \N_N  \, , 
\ee
where 
\be
S_0 = \int_{x_0}^{x_1} \sqrt {V_N (x) - v_{min}} \, dx >0
\ee
is the Agmon distance between two wells and $\lambda_D$ is the ground state of the single well operator $-\epsilon^2 \partial_{xx}^2 + v$, where 
the single well potential $v$ has been introduced by Hyp. 1. \ The numerical pre-factor $\beta$ is the hopping matrix element between neighboring wells, and 
it is such that $4\beta $ is asymptotic to the 
amplitude of the first band of the periodic Bloch operator $H_B$; i.e. $4\beta \sim B_1 := E_1^t - E_1^b$ where $E_1^b$ 
and $E_1^t$ are, respectively, the bottom and the top of the first band. \ Such a numerical pre-factor is going to be exponentially small, i.e. 
\be
\beta = {\asy} (e^{-S_0/\epsilon }) \ \mbox { as } \ \epsilon \to 0^+ \, . 
\ee

\begin {remark} \label {Remark3}
Hyp. \ref {Ipo3} means that, from a practical point of view, the parameter $f$ cannot be arbitrarily small, but it has 
a lower bound of order $\beta$. \ On the other hand, the parameter $\gamma$ may be arbitrarily small.
\end {remark}

The associated normalized eigenvectors are given by \cite {S2}
\be
v_\ell = \sum_{j=1}^N \alpha_{\ell ,j} u_j^{sc} + O(\epsilon^\infty ) e^{-S_0/\epsilon }
\ee
where 
\be
\alpha_{j,\ell }= \alpha_{\ell ,j} = \sqrt {\frac {2}{N+1}} \sin \left ( j \ell \frac {\pi}{N+1} \right ) 
\ee
and where $u_j^{sc} (x)$ is the semiclassical single well ground state eigenfunction localized on the $j$-th cell; by construction and since $v(x)$ is an even function then 
\bee
u_j^{sc} (x) = u_0^{sc} (x- x_j) \ \mbox { and } \ u_0^{sc} (x) = u_0^{sc} (-x) \, . \label {Eq6}
\eee

Now, let $\Pi$ be the projection operator associated with the $N$ eigenvalues $\lambda_\ell $, i.e. 
\be
\Pi = \sum_{\ell =1}^N \langle v_\ell , \cdot \rangle v_\ell 
\ee
and let 
\be
\Pi_c = 1- \Pi \, . 
\ee
Let $F = \Pi (L^2 (\R ))$ be the $N$-dimensional space spanned by the $N$ eigenvectors $v_\ell $, $\ell \in \N_N$.

\begin {remark} \label {Remark4}
Let $\sigma (H_N)$ be the spectrum of $H_N$; then it is a well known semiclassical result that
\be
C^{-1} \epsilon \le 
\mbox {\rm dist} \left ( \left \{ \lambda_\ell \right \}_{\ell=1}^N , \sigma (H_N ) \setminus \left \{ \lambda_\ell \right \}_{\ell=1}^N \right ) \le C \epsilon 
\ee
for some positive constant $C>0$. \ Hence, since $H_N$ is a self-adjoint operator then 
\be
\left \| [H_N - \lambda_D ]^{-1} \Pi_c \right \|_{{\mathcal L} (L^2 \to L^2 )} \le C \epsilon^{-1}
\ee
for some $C>0$.
\end {remark}

\begin {remark} \label {Remark5}
By \cite {FS} it has been proved that there exists a suitable orthonormal base $u_\ell$, $\ell \in \N_N$, of the space 
$F$. \ The functions $u_\ell$ are practically  localized on the $\ell$-th well. \ More precisely, they are such that 

\begin {itemize}
 
\item [i.] $\| u_\ell - u_\ell^{sc} \|_{L^p} = \asy \left ( e^{-S_0/\epsilon } \right ) $ for any $p \in [2,+\infty ]$ and any $\ell \in \N_N $;

\item [ii.] $\|  u_\ell u_j \|_{L^1} = \asy \left ( e^{-S_0|j-\ell |/\epsilon } \right ) $ for any $j , \ell \in \N_N$;

\item [iii.] $ \| u_\ell \|_{L^p} \le C \epsilon^{-\frac {p-2}{4p} } $, $p\in [2,+\infty ]$, and $\| \partial_x u_\ell \|_{L^2} \le 
C \epsilon^{-1/2}$ for any $ \ell \in \N_N$;

\item [iv.] The matrix with elements $\langle u_\ell , H_N u_j \rangle $ can be written as
\be
\left ( \langle u_\ell , H_N u_j \rangle \right ) = \lambda_D \I_N - \beta {\mathcal T} + D_N 
\ee
where ${\mathcal T}$ is the tridiagonal Toeplix matrix such that 
\be
{\mathcal T}_{j,\ell } = 
\left \{
\begin {array}{ll}
0 & \mbox { if } \ |j-\ell |\not= 1 \\ 
1 & \mbox { if } \ |j-\ell | = 1 
\end {array}
\right.
\ee
and where the remainder term $D_N$ is a bounded linear operator from $\ell^p (\N_N) $ to $\ell^p (\N_N) $ with bound
\be
\| D_N \|_{{\mathcal L} (\ell^p (\N_N) \to \ell^p (\N_N) )} = \asy \left ( e^{- (S_0 + \alpha ) /\epsilon } \right )\, , \ p\in [1,+\infty ] \, ,
\ee
for some $\alpha >0$. 

\end {itemize}

\end {remark}

We finally assume that the initial state is prepared on the first $N$ ``ground states''. \ That is

\begin {hypothesis} \label {Ipo4}
{$\Pi_c \psi_0 =0$.}
\end {hypothesis}

It is well known that under the assumptions above the NLS (\ref {Eq3}) is locally well posed, and the conservation of the norm and of the energy \cite {CW,C}
\be
{\mathcal E} (\psi ) = \langle \psi , H_N \psi \rangle + \frac 12 \gamma \| \psi \|_{L^4}^4 + f \langle \psi , W_N \psi \rangle 
\ee
easily follow:
\be
\| \psi (\cdot , t) \|_{L^2} = \| \psi_0 (\cdot ) \|_{L^2} \ \mbox { and } \ 
{\mathcal E} \left ( \psi (\cdot ,t ) \right ) = {\mathcal E} \left ( \psi_0 (\cdot ) \right ) \, . 
\ee
Furthermore the following \emph {a priori} estimate follows, too.

\begin {lemma} \label {lemma1}
There exists a positive constant $C>0$ such that
\be
\| \psi \|_{H^1} \le C \epsilon^{-1/2} \ \mbox { and } \ \| \psi \|_{L^p}^p \le C \epsilon^{-\frac {p-2}{4}} \, , \forall p \in [2,+\infty ]\, . 
\ee
\end {lemma}

\begin {proof}
Indeed, from Theorem 2 by \cite {S1} and its remarks it follows that 
\be
\| \nabla \psi \|_{L^2} \le C \sqrt {\Lambda } \ \mbox { and } \ \| \psi \|_{L^p} \le C \Lambda^{\frac {p-2}{4p}} 
\ee
for some $C>0$ and $\epsilon$ small enough, where 
\be
\Lambda = \frac {\mathcal {E}(\psi_0)-V_{min}}{\epsilon^2} 
\ee
and where $V_{min} = \min_x [V_N (x) + f W_N (x)]$. \ In particular, since $f W_N (x) \ge - f L = O (\epsilon^s )$ for some $s > 2$, because $L$ is 
fixed, and since $\Pi_c \psi_0 =0$ then $\Lambda \sim \epsilon^{-1}$; therefore 
\be
\| \nabla \psi \|_{L^2} \le C \epsilon^{-1/2} \ \mbox { and } \ \| \psi \|_{L^p} \le C \epsilon^{- \frac {p-2}{4p}} \, .
\ee
\end {proof}

Hence, the global well-posedness of the NLS follows \cite {CW,C}.

\subsection {N-mode approximation}

Let $\psi$ be the normalized solution of the NLS equation written in the formula
\bee
\psi = \psi_1 + \psi_c \, ,\ \psi_1 = \Pi \psi = \sum_{\ell =1}^N c_\ell u_\ell \ \mbox { and } \ \psi_c = \Pi_c \psi \, ,  \label {Eq8}
\eee
for some complex-valued functions $c_\ell (t)$. \ By substituting (\ref {Eq8}) into the NLS (\ref {Eq3}) then it takes the formula
\bee
\left \{ 
\begin {array}{ll}
i \epsilon \partial_t c_\ell = \langle u_\ell , H_N \psi_1 \rangle + \gamma \langle u_\ell , |\psi (\cdot , \tau )|^2 \psi (\cdot ,\tau ) \rangle + f 
\langle u_\ell , W_N \psi (\cdot , \tau ) \rangle \\ 
i \epsilon \partial_t \psi_c = H_N \psi_c + \gamma \Pi_c  |\psi (\cdot , \tau )|^2 \psi (\cdot ,\tau ) + f \Pi_c W_N \psi (\cdot , \tau ) 
\end {array}
\right. \label {Eq9}
\eee

We are going now to get an \emph {a priori} estimate of the remainder term $\psi_c$. \ First of all we rescale the time $t \to \tau = \frac {\beta}{\epsilon} t$ 
and we redefine the wavefunction up to a gauge factor $\psi (x,t) \to \psi (x,\tau) := e^{-i \lambda_D t/\epsilon} \psi (x,t)$. \ The Bloch period becomes
\be
\tau_B =  \frac {\beta}{\epsilon} T = \frac {2\pi \beta}{|f| b}
\ee
Hence, (\ref {Eq9}) becomes (where ${}'$ denotes the derivative with respect to $\tau$)
\bee
\left \{ 
\begin {array}{ll}
i \beta c_\ell' = \langle u_\ell , (H_N -\lambda_D) \psi_1 \rangle  + \gamma \langle u_\ell , |\psi (\cdot , t )|^2 \psi (\cdot ,t) \rangle + f 
\langle u_\ell , W_N \psi (\cdot , t) \rangle \\ 
i \beta \psi_c' = (H_N -\lambda_D) \psi_c + \gamma \Pi_c  |\psi (\cdot , t )|^2 \psi (\cdot ,t) + f \Pi_c W_N \psi (\cdot , t) 
\end {array}
\right. \label {Eq10}
\eee

\begin {theorem} \label {Teo1} {Let Hyp.1-4 be satisfied; then it follows that the remainder $\psi_c$ can be estimated for times of order of the Bloch period. \ That is 
for any fixed $M \in \N$ it follows that 
\be
\max_{\tau \in [0,M \tau_B ]} \| \psi_c (\cdot, \tau ) \|_{L^2} \le C  \frac {|f|}{\epsilon} 
\ee
for some positive constant $C>0$.}
\end {theorem}

\begin {proof} From the first equation of (\ref {Eq10}) and recalling that 
\be
\sum_{\ell =1}^N |c_\ell (\tau )|^2 =  \| \psi_1 \|_{L^2}^2 = 1- \| \psi_c \|_{L^2}^2 \le 1 
\ee
then \emph {a priori} estimate follows 
\bee
|c_\ell ' | &\le & 
\frac {\langle u_\ell , (H_N -\lambda_D) \psi_1 \rangle}{\beta} + \frac {|\gamma |}{\beta}  \| \psi \|_{L^\infty}^2 + 
\frac {|f|}{\beta} \| W_N \|_{L^\infty} \nonumber \\  
&\le & C + \frac {|\gamma | \epsilon^{-1/2}}{\beta} + \frac {|f|}{\beta} L \label {Eq11}
\eee
because $\| u_\ell \|_{L^2} =1$ and $\| \psi \|_{L^2} =1$, and from Remark \ref {Remark5} iv. and Lemma \ref {lemma1}.

Concerning $\psi_c $ it satisfies to the following integral equation
\be
\psi_c = I + II 
\ee
where we set 
\be
I &:=&  -\frac {i}{\beta} \int_0^\tau e^{-i(H_N-\lambda_D) (\tau -s)/\beta } \Pi_c A ds \\
II &:=&  -\frac {i}{\beta} \int_0^\tau e^{-i(H_N-\lambda_D) (\tau -s)/\beta } \Pi_c B ds
\ee
and where $A$ and $B$ are defined as
\be
A &:=& \gamma  |\psi_1 |^2 \psi_1 + f W_N \psi_1 \\ 
B & := & \gamma \left [ \bar \psi_1 \psi_c^2 +  |\psi_c |^2 \psi_c +  2|\psi_1 |^2 \psi_c +  2 |\psi_c|^2 \psi_1 +\psi_1^2 \bar \psi_c \right ] + f W_N \psi_c
\ee
such that 
\be
A+B = \gamma |\psi |^2 \psi + f W_N \psi \, . 
\ee
By means of standard arguments \cite {S1} and making use of the fact that the operator $W_N$ is bounded then it 
follows that

\begin {lemma} \label {lemma2} {Let
\be
\Gamma = |\gamma | \epsilon^{-1/2} + |f| \, .
\ee
Then the functions $A$ and $B$ are such that
\be
\| A \|_{L^2} \le C \Gamma \, , \ \| B \|_{L^2} \le C \Gamma \| \psi_c \|_{L^2} \ \mbox { and } \ \left \| \frac {\partial A}{\partial \tau } \right \|_{L^2} \le C \Gamma^2 
\beta^{-1}  \, .
\ee
}
\end {lemma}

\begin {proof}
Indeed, 
\be
\| A\|_{L^2} \le |\gamma |\, \| \psi_1 \|_{L^\infty}^2 \| \psi_1 \|_{L^2} + |f| \| W_N \|_{L^\infty} \| \psi_1 \|_{L^2} \le C 
\left [ |\gamma | \epsilon^{-1/2} + |f| \right ] 
\ee
since $\| \psi_1 \|_{L^\infty} \le C \max_{\ell} \| u_\ell \|_{L^\infty} \le C \epsilon^{-1/4}$.  \ Similarly, the estimate of the function $B$ 
follows recalling that $\| \psi_c \|_{L^\infty} \le C \epsilon^{-1/4}$ from Lemma \ref {lemma1}. \ Finally, the estimate concerning 
$\frac {\partial A}{\partial \tau }$ immediately follows from (\ref {Eq11}).
\end {proof}

Hence, the estimates of the integrals I and II follow; in particular, integral II can be simply estimated as
\be
\| II \|_{L^2} \le C \Gamma \beta^{-1} \int_0^\tau  \| \psi_c (\cdot , s ) \|_{L^2} ds 
\ee
since 
\be
\left \| e^{-i (H_N - \lambda_D) (\tau -s)/\beta } \right \|_{{\mathcal L} (L^2 \to L^2)} = 1 \, . 
\ee
On the other hand, before to get the estimate of integral I we perform an integration by parts in order to gain a pre-factor $\beta$:
\be
I &=& \left [ - i e^{-i (H_N - \lambda_D) (\tau -s)/\beta } [H_N - \lambda_D ]^{-1} \Pi_c A \right ]_0^\tau + \\ 
&& \ \ + i  
\int_0^\tau e^{-i (H_N - \lambda_D) (\tau -s)/\beta } [H_N - \lambda_D ]^{-1} \Pi_c \frac {\partial A}{\partial s} ds
\ee
From this fact and recalling that (Remark 4)
\be
\| [H_N - \lambda_D ]^{-1} \Pi_c \|_{{\mathcal L} (L^2 \to L^2)} \le C \epsilon^{-1} 
\ee
then 
\be
\| I \|_{L^2} \le C \epsilon^{-1} \max_{s \in [0,\tau ]} \left [ \| A\|_{L^2} + \tau \left \| \frac {\partial A}{\partial s} \right \|_{L^2} \right ] 
\le C \epsilon^{-1} \Gamma [1+ \Gamma \beta^{-1} \tau ] \, .  
\ee
Therefore, we have that
\be
\| \psi_c \|_{L^2} \le C \Gamma \beta^{-1} \int_0^\tau \| \psi_c (\cdot , s) \|_{L^2} \, ds + C \epsilon^{-1} \Gamma ( 1+ \Gamma \beta^{-1}  \tau )\, . 
\ee
From the Gronwall's Lemma it follows that
\be
\| \psi (\cdot ,\tau ) \|_{L^2} \le C \epsilon^{-1} \Gamma ( 1+ \Gamma \beta^{-1}  \tau ) e^{C \Gamma \beta^{-1} \tau } 
\ee

In particular we observe that 
\be
\max_{\tau \in [0,M\tau_B]} \| \psi (\cdot ,\tau ) \|_{L^2} 
\le  C \epsilon^{-1} \Gamma ( 1+ \beta^{-1} \Gamma M \tau_B ) e^{C \Gamma \beta^{-1} M \tau_B } \le C \Gamma \epsilon^{-1} \le C \frac {|f|}{\epsilon} 
\ee
proving the Theorem since $\Gamma \le C |f|$ from Hyp. 3 and since $\tau_B = \frac {2\pi}{b} \frac {\beta }{|f|}$. \end {proof}

We are going now to estimate the solutions $c_\ell$ of the first equation of (\ref {Eq10}) which can be written as
\be
i \beta c_\ell' = \langle u_\ell , (H_N -\lambda_D) \psi_1 \rangle  +  \langle u_\ell , A \rangle + \langle u_\ell , B \rangle 
\ee
where the term $\langle u_\ell , (H_N -\lambda_D) \psi_1 \rangle$ can be represented by property iv. of Remark \ref {Remark5}. \ Concerning the term 
$\langle u_\ell , B \rangle$ the following estimate uniformly holds with respect to the index $\ell$
\be
\left | \langle u_\ell , B \rangle \right | \le \| B \|_{L^2} \le C \Gamma \| \psi_c \|_{L^2} \, . 
\ee
Furthermore
\be
\langle u_\ell , A \rangle &=& \gamma \sum_{j,k,m=1}^N \bar c_j c_k c_m \langle u_\ell ,  \bar u_j u_k u_m \rangle + f \sum_{j=1}^N c_j \langle u_\ell ,   W_N u_j \rangle \\ 
&=& \gamma |c_\ell |^2 c_\ell \| u_\ell \|_{L^4}^4 + f c_\ell \langle u_\ell , W_N u_\ell \rangle + \gamma r_\ell^a + f r_\ell^b 
\ee
where
\be
r_\ell^a = \sum_{j,k,m \in \N_N \ : \ |j-\ell |+|m-\ell |+|k-\ell |>0} \bar c_j c_k c_m \langle u_\ell ,  \bar u_j u_k u_m \rangle 
\ee
and 
\be
r_\ell^b = 
\sum_{j,\ell \in \N_N\, , \ j\not= \ell} c_j \langle u_\ell , W_N u_j \rangle
\ee
are remainder terms. 

We have that

\begin {lemma} \label {lemma3} 
The following estimates uniformly hold with respect to the indexes $\ell \, ,\ j \, ,\  m$ and $k$:

\begin {itemize}

\item [i.] $\langle u_\ell , W_N u_\ell \rangle = \ell b + \asy \left (e^{-S_0/\epsilon } \right ) $;

\item [ii.] $\langle u_\ell , W_N u_j \rangle = \asy \left (e^{-S_0 |j-\ell |/\epsilon } \right ) $;

\item [iii.] $\langle u_\ell , \bar u_j u_m u_k \rangle = \asy \left (e^{-S_0 r /\epsilon } \right ) $ where 
\be
r = \max \left [ |j-\ell |, |m-\ell |, |k-\ell |, |j-m |, |j-k |, |k-m |  \right ] .
\ee

\end {itemize}

\end {lemma}

\begin {proof}
Indeed, let $I_\ell = \left [ x_\ell -  b, x_\ell +  b \right ]$, then 
\be
\langle u_\ell , W_N u_\ell \rangle = \int_{I_\ell} |u_\ell (x)|^2 x dx + \int_{\R \setminus I_\ell} |u_\ell (x)|^2 W_N (x) dx
\ee
where $u_\ell (x) = u_0 (x- x_\ell ) + \asy \left (e^{-S_0 /\epsilon } \right )$ from (\ref {Eq6}) and Remark 4. \ Therefore  
\be
\int_{I_\ell} |u_\ell (x)|^2 x dx &=& \ell b \int_{I_0} |u_0 (x)|^2  dx + \int_{I_0} |u_0 (x)|^2 x dx + \asy \left (e^{-S_0 /\epsilon } \right ) \\ 
&=& \ell b \left [ \| u_0 \|^2_{L^2} - \| u_0 \|^2_{L^2 (\R \setminus I_0)} \right ] + \asy \left (e^{-S_0 /\epsilon } \right )  
\ee
where $u_0$ is normalized and $\int_{I_0} |u_0 (x)|^2 x dx = \asy \left (e^{-S_0/\epsilon } \right )$ because $u_0 (x)= u_0 (-x) +\asy \left 
(e^{-S_0 /\epsilon } \right )$. \ From this fact and since $\| u_\ell \|_{L^2 (\R \setminus I_\ell )} = \asy (e^{-S_0/\hbar }) $ 
(see Lemma 4 iii. and Lemma 5 by \cite {FS}) then the asymptotic behavior i. follows. \ The other two asymptotic behaviors ii. and iii. similarly 
follow from property ii. by Remark \ref {Remark5}; indeed 
\be
| \langle u_\ell , W_N u_j \rangle | \le \| W_N \|_{L^\infty} \| u_\ell u_j \|_{L^1} = \asy \left (e^{-S_0 |j-\ell |/\epsilon } \right )
\ee
and, where we assume that $r = |j-\ell |$, 
\be
| \langle u_\ell , \bar u_j u_m u_k \rangle | \le \| u_m\|_{L^\infty} \| u\|_{L^\infty} \| u_\ell u_j \|_{L^1} = \asy \left (e^{-S_0 r  /\epsilon } \right )
\ee
proving so the estimates ii. and iii..
\end {proof}

From this Lemma and from the previous computation it follows that the first equation of (\ref {Eq10}) becomes a DNLS of the form 
\bee
i \beta c_\ell' = - \beta \sum_{m=1}^N {\mathcal T}_{\ell , m} c_m + \gamma \| u_\ell \|_{L^4}^4 |c_\ell |^2 c_\ell + f \ell b c_\ell +  
\asy \left (e^{-S_0 /\epsilon } \right ) \label {Eq15}
\eee
where 
\be
\| u_\ell \|_{L^4}^4 = \| u_0 \|_{L^4}^4 +  \asy \left (e^{-S_0 /\epsilon } \right )
\ee
and where the remainder terms $ \asy \left (e^{-S_0 /\epsilon } \right )$ are uniform with respect to the index $\ell$.

\begin {remark} \label {Remark6}
In fact, if $v(x)$ is not an even function then by means of standard semiclassical arguments it follows that property Lemma 3 i. becomes 
\be
\langle u_\ell , W_N u_\ell \rangle = \ell b +c  + \asy \left (e^{-S_0/\epsilon } \right ) 
\ee
for some constant $c$ independent of the index $\ell$. \ In such a case we must add the term $f c c_\ell$ to the right hand side of the DNLS above and, by means of 
a gauge choice $c_\ell \to c_\ell e^{-i \frac {f c}{\beta }\tau }$, we can remove this term obtaining again equation (\ref {Eq15}). 
\end {remark}

Now, we are able to prove that

\begin {theorem} \label {Teo2}
Let $d_\ell (\tau )$ be the solutions of the DNLS
\bee
i \beta d_\ell' = - \beta \sum_{m=1}^N {\mathcal T}_{\ell , m} d_m + \gamma \| u_0^{sc} \|_{L^4}^4 |d_\ell |^2 d_\ell + f \ell b d_\ell \label {Eq16}
\eee
satisfying to the initial conditions $d_\ell (0) = c_\ell (0)$, where $c_\ell (\tau )$ and $\psi_c$ are the solutions of (\ref {Eq10}). \ Then, for any fixed $M \in \N$ it 
follows that
\be
\max_{\tau \in [0,M\tau_B] ,\, \ell =1,2, \ldots ,N} | c_\ell (\tau ) - d_\ell (\tau )| = \asy \left (e^{-S_0 /\epsilon } \right ) \ 
\mbox { as } \ \epsilon \to 0 \, . 
\ee
\end {theorem}

\begin {proof}
The proof is a simply consequence of equation (\ref {Eq15}) and from the fact that $\tau_B = \frac {2\pi }{b} \frac {\beta}{f} $ and Hyp.3. 
\end {proof}

\section {Numerical analysis of a real model}

We consider the experiment where a cloud of ultracold Strontium atoms ${}^{88} Sr$ are trapped in a one-dimensional optical lattice with potential 
(\ref {Eq1}). \ Realistic data for the experiment are \cite {PWTAPT}:

\begin {itemize}

 \item [-] Lattice period: $b=\lambda_L/2 = 266 \, nm$, $\lambda_L = 532 \, nm$;

 \item [-] Lattice potential depth: $V_0 = \Lambda_0 \cdot E_R$ where $E_R$ is the photon recoil energy $E_R = \frac {2 \pi^2 \hbar^2}{m\lambda_L^2} = 
 50.38\, kHz \cdot \hbar$ and where $\Lambda_0$ is between $3$ and $10$;
 
 \item [-] Mass of the strontium $88$ isotope: $m= 87.91 \, au = 1.46 \cdot 10^{-22} \, gr$;
 
 \item [-] Effective one-dimensional nonlinearity strength: let $\gamma_{3D} = \frac {4{\mathcal N} \pi a_s \hbar^2}{m}$ be the effective nonlinearity strength for the 
 three-dimensional Gross-Pitaevskii equation, then it is expected that the effective one-dimensional nonlinearity strength $\gamma$ is of the order \cite {SPR}
 \be
 \gamma \approx \frac {\gamma_{3D}}{2 \pi d_\perp^2} 
 \ee
 where $d_\perp$ is the oscillator length of the transverse confinement; here $a_s$ denotes the scattering length of the Strontium $88$ isotope: 
 $a_s = -a_0 \div 13 a_0$, where $a_0$ is the Bohr radius; ${\mathcal N}$ is the number of atoms of the condensate; in typical experiments 
$d_\perp \approx 180 \cdot 10^{-6}\, m$ and ${\mathcal N} = 10^5 \div 10^6$;

\item [-] Acceleration constant $ g = 9.807 \, m/s^2$.
 
\end {itemize}

The confined BEC is governed by Eq. (\ref {Eq3}) and here we make use of the $N$-mode approximation (\ref {Eq16}), that is the wavefunction $\psi$ 
has the form $\psi  \sim \sum_\ell c_\ell  u_\ell $ where $c_\ell$ are the solutions of (\ref {Eq16}) and where $u_\ell$ are 
functions localized on the $\ell$-th lattice site. \ In order to justify the validity of such an approximation we'll 
check if the model is in the \emph {semiclassical regime}, that is if the first band is almost flat and if semi-classical approximation $u_\ell^{sc}$ agrees 
or not with the Wannier function $u_\ell^W$. \ Such a qualitative criterion has been also adopted by other authors \cite {AKKS,BGHH,EHLZCMA} and we'll see 
that our results agree with the results contained in these papers. \ In particular, in \cite {EHLZCMA} has been computed the hopping matrix elements 
$\langle u_\ell , H_N u_j \rangle $ too, where it has been numerically verified that these coefficients are negligible when $|j- \ell |>1$ for 
$\Lambda_0 \ge 10$; thus, for such values of $\Lambda_0$ it is admitted that the $N$-mode approximation, consisting to describe (\ref {Eq3}) in terms of a 
nearest-neighbor model (\ref {Eq16}), works. 

\subsection {Validity of the semiclassical approximation}

The semiclassical approximation $u_0^{sc}(x)$ of the wavefunction has dominant behavior 
\bee
u_0^{sc} (x) = \frac {(m \mu )^{1/8}}{(\pi \hbar )^{1/4}} e^{-\sqrt {m \mu} x^2/2\hbar} \label {Eq19}
\eee
in the semiclassical limit, where $\mu = \frac {d^2 V_{per} (0)}{dx^2} = 2 V_0 k_L^2$, $V_0 = \Lambda_0 E_R$; it is normalized $\| u_0^{sc} \|_{L^2} =1$. \ We may remark that the 
effective semiclassical parameter in adimensional units is given by 
\be
\frac {1}{\sqrt {\Lambda_0}} = \frac {2 \pi^2 \hbar}{b^2 \sqrt {m \mu}} \, ,
\ee
and then the semiclassical approximation may be written as 
\be
u_0^{sc} (x) = \left [ \frac {2\pi \sqrt {\Lambda_0}}{b^2} \right ]^{1/4} e^{-x^2 \pi^2 \sqrt {\Lambda_0}/b^2} \, . 
\ee
Hence
\be
\| u_0^{sc} \|_{L^4}^{4} = \left [ \frac {m \mu }{(\pi \hbar )^2} \right ]^{1 /4} \sqrt {\frac {\pi}{2}} = \frac {\pi \Lambda_0^{1/4}}{b} .
\ee
We'll see that for $\Lambda_0$ ``large enough'' (i.e. $\Lambda_0 \ge 10$) then the first band is almost flat and the 
semiclassical function $u_0^{sc}$ well approximates the Wannier function $u_0^W$, as we expect to observe in the 
semiclassical limit $\Lambda_0 \to \infty$ (see, e.g., \cite {WK}).

\begin {remark} \label {Remark7}
By the scaling
\be
x \to  2 k_L x\, , \ t \to  E_R t/\hbar \, , \psi (x) \to \frac {1}{\sqrt {2 k_L}} \psi \left ( \frac {x}{2 k_L} \right ) 
\ee
and setting
\be
F = \frac {mg}{2E_R k_L} \, , \ \zeta = \frac {\gamma}{2E_R k_L} \, , \ \epsilon = \frac {1}{\sqrt {\Lambda_0}} 
\ee
then (\ref {Eq2}) takes the form
\bee
i \partial_t \psi = - \partial^2_{xx} \psi + \frac {1}{\epsilon^2 } \sin^2 (x/2) + F x \psi + \zeta |\psi |^2 \psi \, , \ \| \psi \|_{L^2} =1 \, . \label {Eq20}
\eee
Equation (\ref {Eq20}) is equivalent, up to a change of scale of the time, to the equation
\be
i \epsilon \partial_t \psi = - \epsilon^2 \partial^2_{xx} \psi + \sin^2 \left ( \frac x2 \right ) \psi + f x \psi + \gamma |\psi |^2 \psi 
\ee
where we set
\be
t \to t/\sqrt {\epsilon} \, , \ f = F \epsilon^2 \, , \ \gamma = \epsilon^2 \zeta \, . 
\ee
and where $\epsilon = \Lambda_0^{-1/2}$ plays the role of the semiclassical parameter.
\end {remark}

We compute now the band functions and the Wannier functions for different values of $\Lambda_0$. \ The semiclassical wavefunction $u_0^{sc}$ is computed by (\ref {Eq19}), 
while the Wannier function $u_0^W (x)$ may be computed by means of the Mathieu functions (see Appendix).

\subsubsection {Model $\Lambda_0 =3$}

The first bands of the Bloch operator $H_B = - \frac {\hbar^2}{2m} \partial^2_{xx} + \Lambda_0 E_R \sin^2 (k_L x) $ have endpoints 

\begin {itemize}
 \item [n=1)] $ {E}_1^b = 1.43 \cdot E_R$ and $ {E}_1^t = 2.11 \cdot E_R$;
 
 \item [n=2)] $ {E}_2^b = 2.86 \cdot E_R$ and $ {E}_2^t = 5.49 \cdot E_R$;
 
 \item [n=3)] $ {E}_3^b = 5.56 \cdot E_R$ and $ {E}_3^t = 10.51 \cdot E_R$.
 
 \end {itemize}
 
 Hence, the values of the width of the first two bands are given by 
 \be
 B_1 := {E}_1^t -  {E}_1^b = 0.68 \cdot E_R \ \ \mbox { and } \ \  B_2 := {E}_2^b -  {E}_1^t = 2.63 \cdot E_R \, . 
 \ee
 Furthermore it follows that the first gap has amplitude $g_1 = E_2^b - E_1^t =0.77 \cdot E_R$ of the order of the first band amplitude, while the width of 
 the other gaps are very small (see Fig. \ref {Figura2}, left hand side panel). \ If we compare the first Wannier 
 function $u_0^{W} (x)$ and the semiclassical approximation $u_0^{sc} (x)$ it turns out that (see also Fig. \ref {Figura2}, right hand side panel)
 \be
 \| u_0^W  - u_0^{sc}  \|_{L^2}^2 = 0.091 
 \ee
\begin{figure} [h]
\begin{center}
\includegraphics[height=6cm,width=6cm]{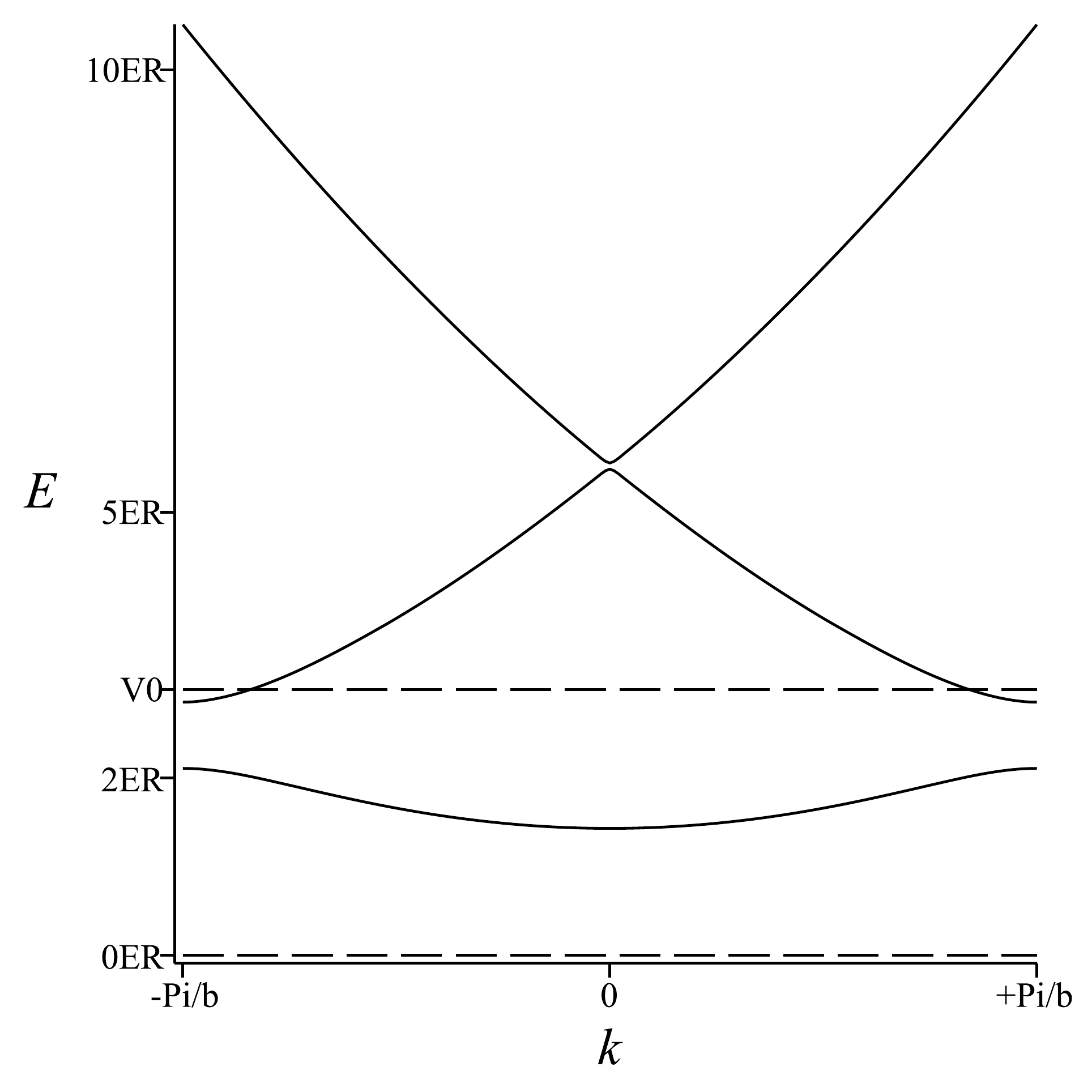}
\includegraphics[height=6cm,width=6cm]{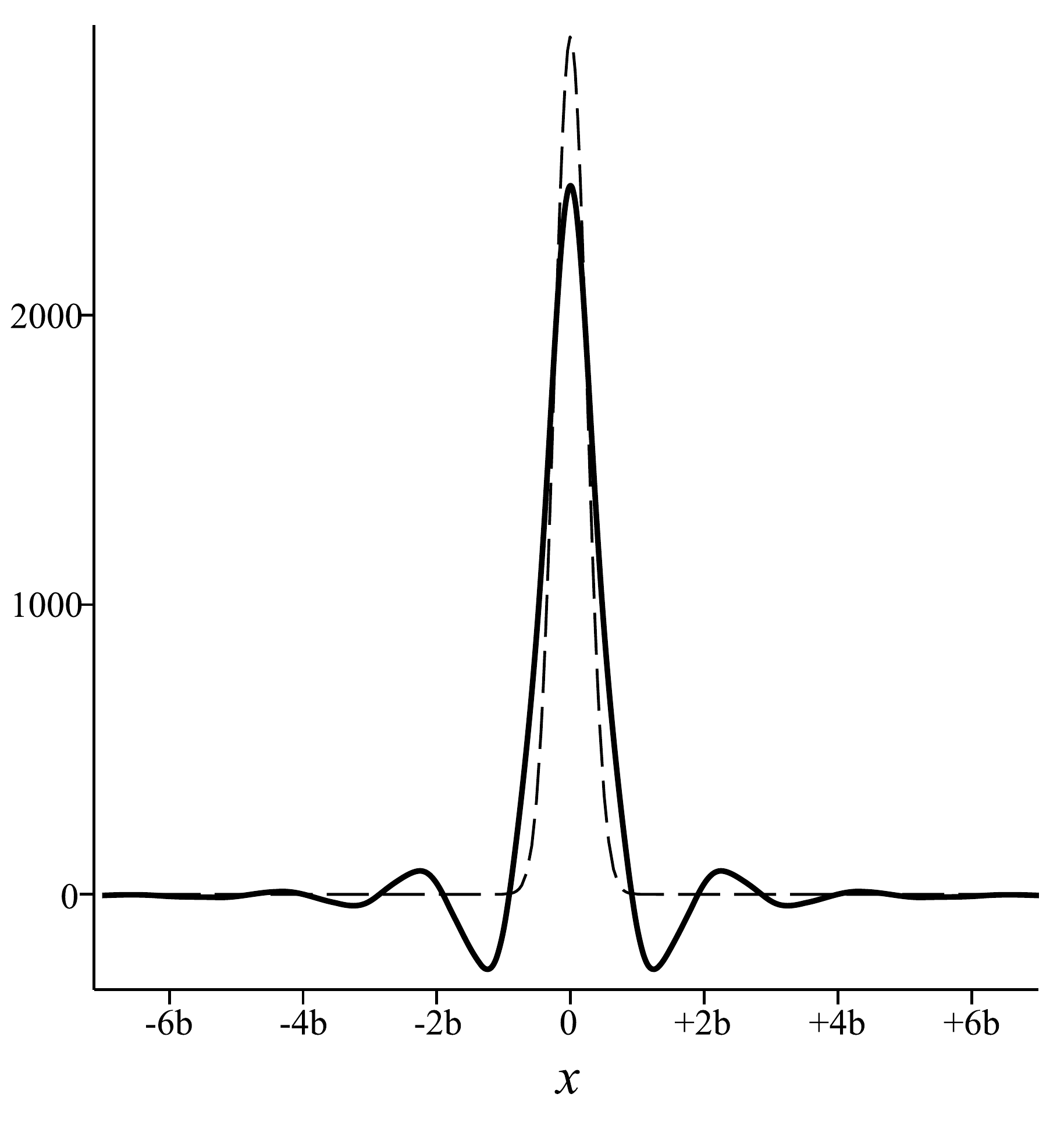}
\caption{\label {Figura2} Here we plot in the left hand side panel the first three band functions $E_n (k)$, $n=1,2,3$ and $k\in \left [ - \frac {\pi}{b} , + \frac {\pi}{b} 
\right ]$, for the Bloch operator $H_B$ where $\Lambda_0 = 3$. \ It turns out that the width of the first gap is of the same order of the width of the first band. \ In the 
right hand side panel we plot the graph of the functions $u_0^{sc}$ (broken line) and $u_0^W $ (full line).}
\end{center}
\end{figure}

\subsubsection {Model $\Lambda_0 =10$}
 
The first bands of the Bloch operator $H_B$ have endpoints 

\begin {itemize}
 \item [n=1)] $ {E}_1^b = 4.32 \cdot E_R$ and $ {E}_1^t = 4.58 \cdot E_R$;
 
 \item [n=2)] $ {E}_2^b = 7.02 \cdot E_R$ and $ {E}_2^t = 8.87 \cdot E_R$;
 
 \item [n=3)] $ {E}_3^b = 9.54 \cdot E_R$ and $ {E}_3^t = 14.07 \cdot E_R$.
 
 \end {itemize}
                                    
 Hence, the values of the width of the first two bands are given by 
 \be
 B_1 := {E}_1^t -  {E}_1^b = 0.26 \cdot E_R \ \ \mbox { and } \ \  B_2 := {E}_2^b -  {E}_1^t = 1.85 \cdot E_R \, . 
 \ee
 Furthermore it also follows that the first gap has amplitude $g_1 = E^b_2 - E_1^t = 2.44 \cdot E_R$ is much larger than the 
amplitude of the first band and that the width of the other gaps are very small (see Fig. \ref {Figura3}, left hand side panel). \ If 
 we compare the first Wannier  function $u^W_0 (x)$ and the semiclassical approximation $u_0^{sc} (x)$ it turns out that 
 (see also Fig. \ref {Figura3}, right hand side panel)
 \be
 \| u_0^W  - u_0^{sc}  \|_{L^2}^2 = 0.055 \, . 
 \ee
Hence, we may conclude that for $\Lambda_0 =10$ the $N$-mode approximation properly works.
\begin{figure} [h]
\begin{center}
\includegraphics[height=6cm,width=6cm]{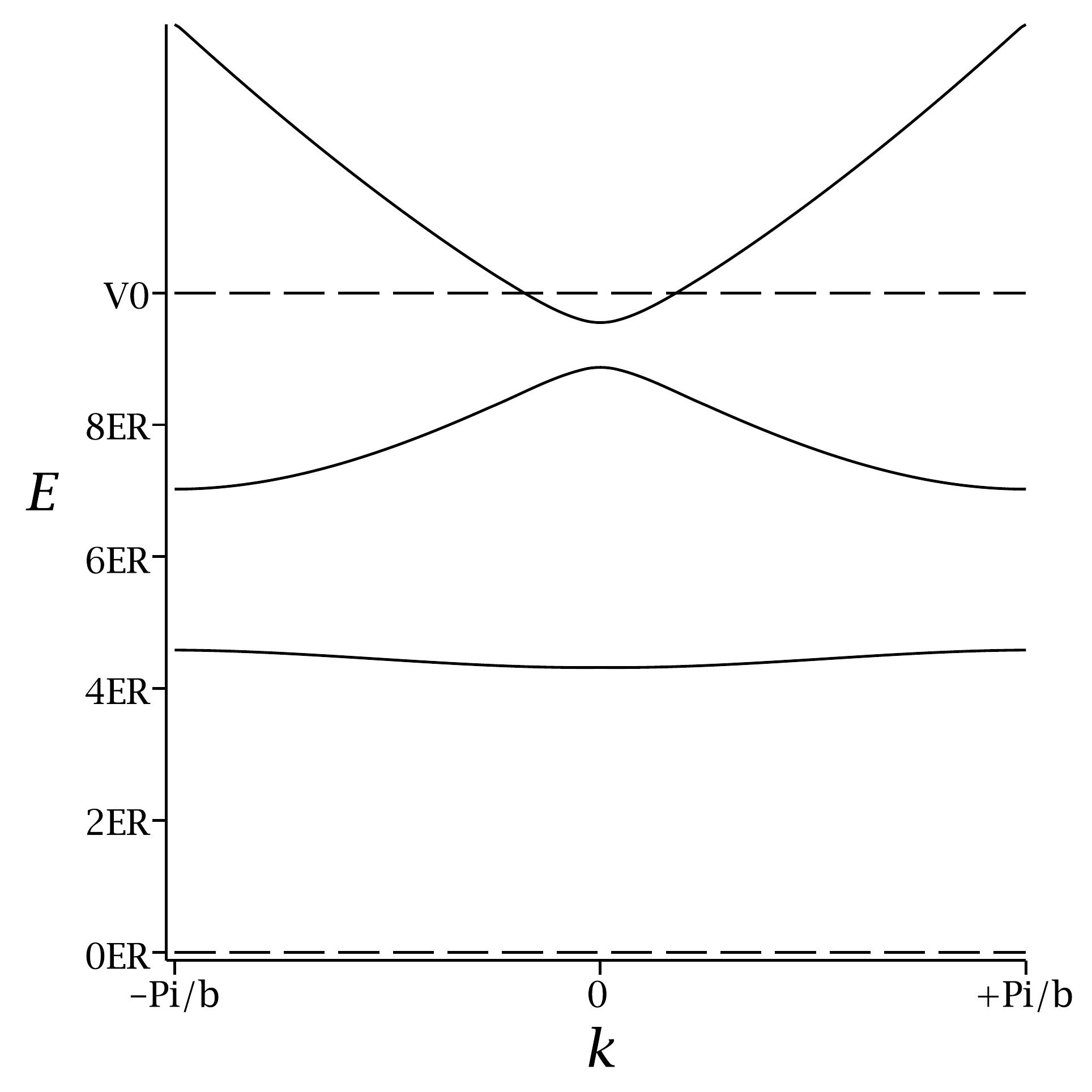}
\includegraphics[height=6cm,width=6cm]{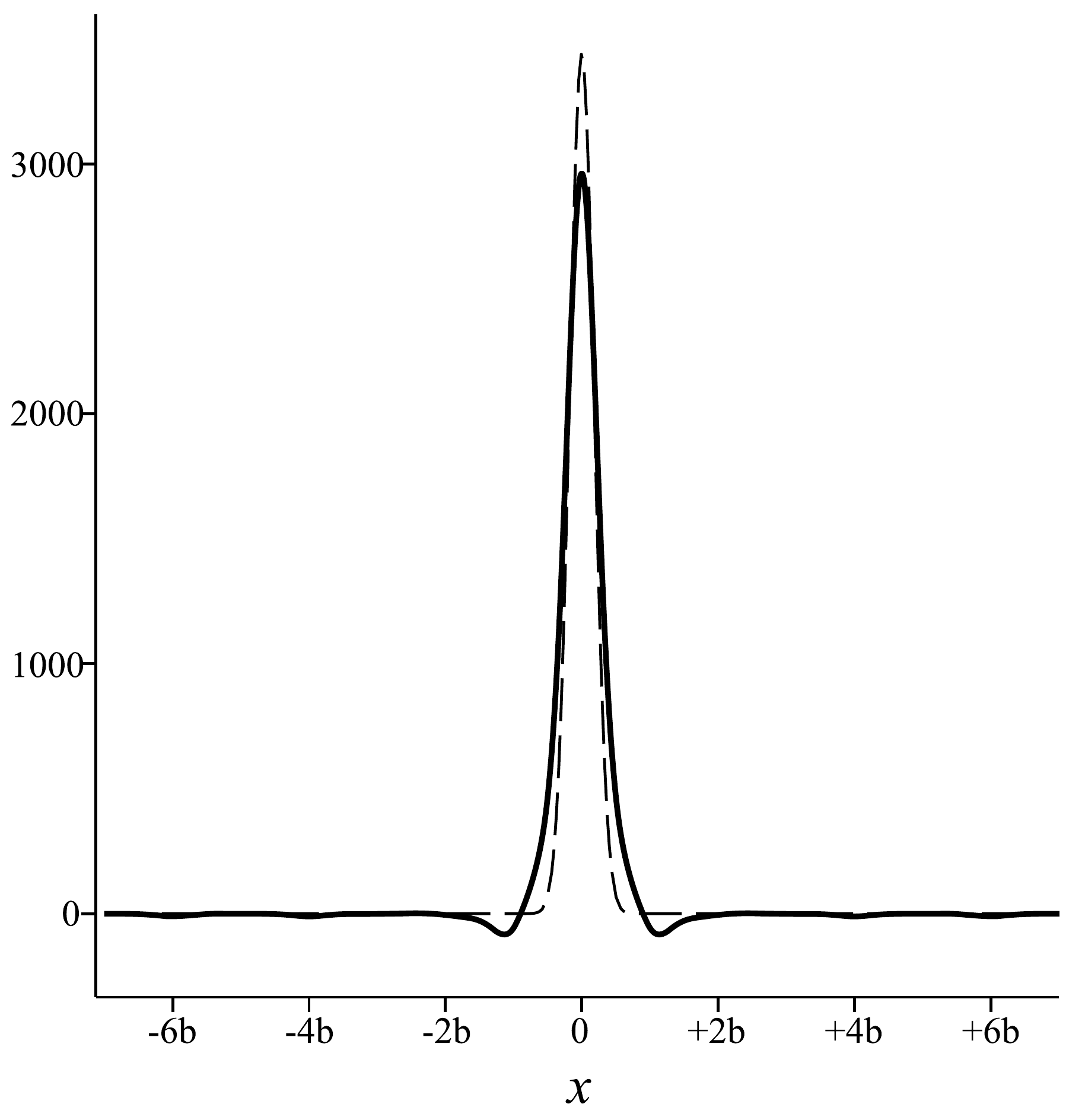}
\caption{\label {Figura3} Here we plot in the left hand side panel the first three band functions $E_n (k)$, $n=1,2,3$ and $k\in \left [ - \frac {\pi}{b} , + 
\frac {\pi}{b} \right ]$, for the Bloch operator $H_B$ where $\Lambda_0 = 10$. \ It turns out that the first band is almost flat, in fact its width is $1/10$-th 
of the width of the first gap. \ In the right hand side panel we plot the graph of the functions $u_0^{sc}$ (broken line) 
and $u_0^W $ (full line).}
\end{center}
\end{figure}

\subsection {Numerical analysis of the model for $\Lambda_0 =10$}

We have seen that for $\Lambda_0 \ge 10$ the $N$-mode approximation is justified. \ For $\Lambda_0 =10$ we have that 
\be
\beta \sim \frac 14 B_1 = 0.065 \cdot E_R \, . 
\ee
Equation (\ref {Eq16}) takes the form
\be
i d_\ell' =  - \sum_{m=1}^N {\mathcal T}_{\ell , m} d_m + \eta |d_\ell |^2 d_\ell +  \ell \delta d_\ell \, , \ \ell \in \N_N \, , 
\ee
where we set
\be
\eta = \frac {\gamma \| u_0^{sc} \|_{L^4}^4}{\beta } \approx \frac {4 {\mathcal N} \pi a_s \hbar^2}{m} \frac {{\pi} \Lambda_0^{1/4} }{b} \frac {1}{2\pi d_\perp^2} 
\frac {1}{\beta} = -0.151 \cdot 10^{-1} \div 0.197
\ee
and 
\be 
\delta = \frac {f b}{\beta} = \frac {m g b}{\beta} = 1.103 \, .
\ee
The Bloch period is given by
\be
T = \frac {2\pi \hbar}{m g b } = 1.740 \, ms \, .
\ee
Hence, the parameters $f$, $\gamma$ and $\beta$ are in a suitable range as discussed in Remark \ref {Remark3}. \ Furthermore, the motion of the Bloch oscillator occurs in 
an interval with width
\bee
\frac {B_1}{|f|} = \frac {0.26 \cdot E_R}{m g } = 9.65 \cdot 10^{-7} m \approx 3.6 \cdot b \, . \label {BlochRange}
\eee
Hence, the $N$-mode approximation with $N=40$ properly works. 

We consider three different situations. \ In the first one we assume that the state is initially prepared on a single lattice site, 
that is $\psi_0$ is a Wannier type function. \ In the other two cases we assume that the initial wavefunction $\psi_0$ is a symmetric or asymmetrical wavefunction 
initially prepared on different lattice sites.

\subsubsection {$\psi_0$ is initially prepared on a single lattice cell} We consider a numerical experiment where $\psi_0 (x) = u_0 (x)$, that is 
$c_\ell (0) =0$, for $\ell \not= N/2$, and $c_{N/2}(0) =1$ (where $N=40$). \ In fact, in such a case we observe a \emph {breathing motion} for the wavefunction; that is, 
the wavefunction, initially prepared in a Wannier state localized on a single site of the optical lattice, symmetrically spreads in space and it 
periodically returns to its initial shape (Fig. \ref {Figura4}, top panel, obtained for $\eta =0.2$). \ Then the expected value of the center of mass  
\be
\langle x \rangle^t = \langle \psi (\cdot ,t) ,x \psi (\cdot , t) \rangle 
\ee
is practical constant $\langle x \rangle^t \approx 0$ up to small fluctuations.
\begin{figure} [h]
\begin{center}
\includegraphics[height=6cm,width=12cm]{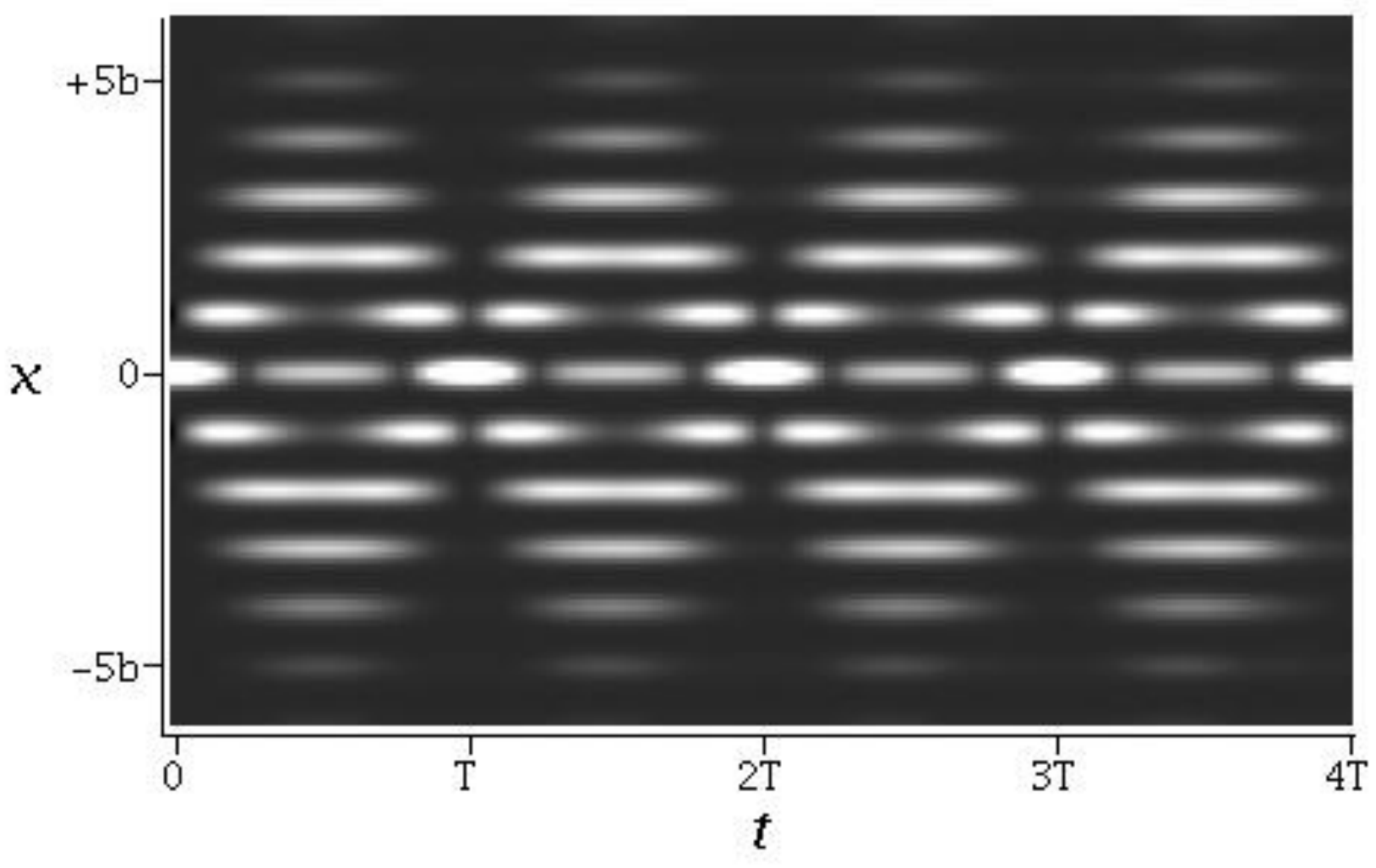}\\
\includegraphics[height=6cm,width=12cm]{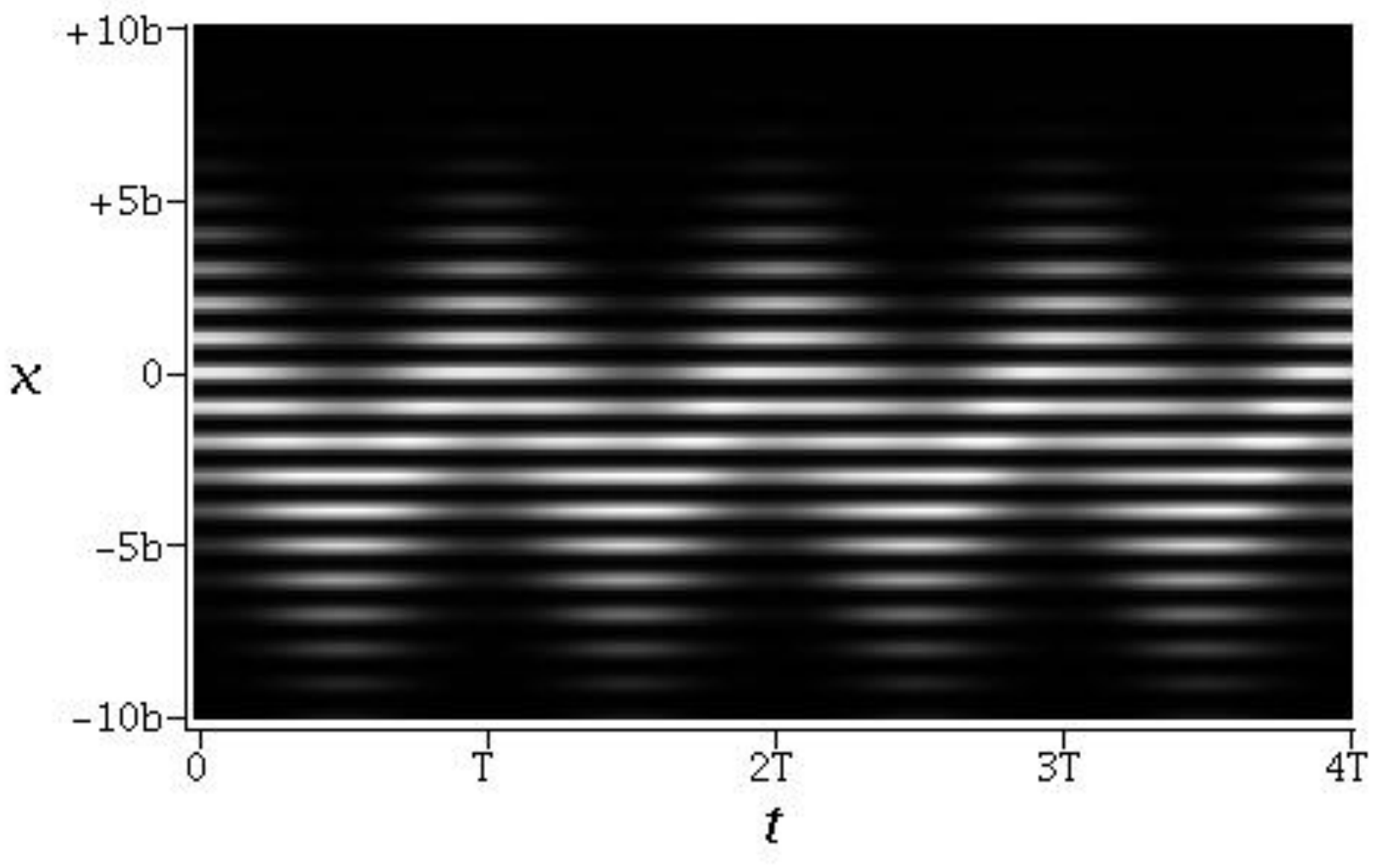}
\caption{\label {Figura4} In the top panel we plot the absolute value of the wavefunction $\psi (x,t)$ initially prepared on a single Wannier state for $\eta =0.2$, 
it turns out that it 
symmetrically spreads in space and periodically returns to its initial shape without motion of the center of mass. \ In the bottom panel we plot the absolute value of 
the wavefunction initially prepared on several lattice sites for $\eta =0.2$; it turns out that the center of mass oscillates with no marked changes of the shape of 
the wavefunction. \ Here $T$ denotes the Bloch period and $b$ is the distance between two adjacent wells. \ Dark regions mean that $|\psi (x,t)|$ is practically zero 
there, white regions mean that $|\psi (x,t)|$ has its maximum value there.}
\end{center}
\end{figure}

\subsubsection {$\psi_0$ is a symmetric wavefunction initially prepared on different lattice cells} We consider a numerical experiment where $N=40$ and 
$\psi_0 (x) = \sum_{\ell=0}^{40} c_\ell u_\ell (x)$, where $c_\ell$ have a symmetric Gaussian-type distribution around $\ell = N/2$. \ That is the initial value of 
the coefficients $c_\ell (t)$ are given in Table \ref {tabella1}, the initial wavefunction $\psi_0$ is plotted in Fig. \ref {Figura5}, left hand side panel. \ In such a case the 
center of mass $\langle x \rangle^t$  oscillates in space and the wavefunction moves with no marked changes in shape (see Fig. \ref {Figura4}, bottom panel). \ In particular, the 
function $\langle x \rangle^t$ exhibits, for $\eta \not= 0$, an oscillating motion where the wavefunction amplitude is modulated (see Fig. \ref {Figura6}) and where the 
oscillating (pseudo-)period (that is the time interval between two consecutive minima or maxima points) depends on $\eta$. \ In Fig. \ref {Figura7} we plot 
the mean value of the oscillating period of the motion of the center of mass after 14 oscillations for $\eta$ in the range $[-0.1,+0.2]$; it turns out that the relative uncertainty 
with respect to the Bloch period is of order $2.4 \cdot 10^{-5}$.
\begin{table}
\begin{center}
\begin{tabular}{|c|c|c|c|} 
\hline
$c_0 =0$                     &$c_{10} =0.396 \cdot 10^{-3} $  &$c_{21}=0.429$                &$c_{31}=0.898 \cdot 10^{-4}$ \\ \hline
$c_1 =0$                     &$c_{11} =0.151 \cdot 10^{-2} $  &$c_{22}=0.347$                &$c_{32}=0.177 \cdot 10^{-4}$ \\ \hline
$c_2 =0$                     &$c_{12} =0.502 \cdot 10^{-2} $  &$c_{23}=0.244$                &$c_{33}=0.303 \cdot 10^{-5}$ \\ \hline
$c_3 =0$                     &$c_{13} =0.149 \cdot 10^{-1} $  &$c_{24}=0.149$                &$c_{34}=0$ \\ \hline
$c_4 =0$                     &$c_{14} =0.363 \cdot 10^{-1} $  &$c_{25}=0.788 \cdot 10^{-1}$  &$c_{35}=0$ \\ \hline
$c_5 =0$                     &$c_{15}= 0.788 \cdot 10^{-1} $  &$c_{26}=0.363 \cdot 10^{-1}$  &$c_{36}=0$ \\ \hline
$c_6 =0$                     &$c_{16}= 0.149               $  &$c_{27}=0.149 \cdot 10^{-1}$  &$c_{37}=0$ \\ \hline
$c_7 =0.303 \cdot 10^{-5} $  &$c_{17}= 0.244               $  &$c_{28}=0.502 \cdot 10^{-2}$  &$c_{38}=0$ \\ \hline
$c_8 =0.177 \cdot 10^{-4} $  &$c_{18}= 0.347               $  &$c_{29}=0.151 \cdot 10^{-2}$  &$c_{39}=0$ \\ \hline
$c_9 =0.898 \cdot 10^{-4} $  &$c_{19}= 0.429               $  &$c_{30}=0.396 \cdot 10^{-3}$  &$c_{40}=0$ \\ \hline 
  &$c_{20}=0.460$  &  & \\ \hline
\end{tabular}
\caption{\small Initial values of the coefficients $c_\ell :=c_\ell (0)$ of the wavefunction. \ The initial wavefunction $\psi_0$ has a symmetric shape 
and its width is of order of several lattice periods.}
\label{tabella1}
\end{center}
\end {table}

\begin{figure} [h]
\begin{center}
\includegraphics[height=6cm,width=6cm]{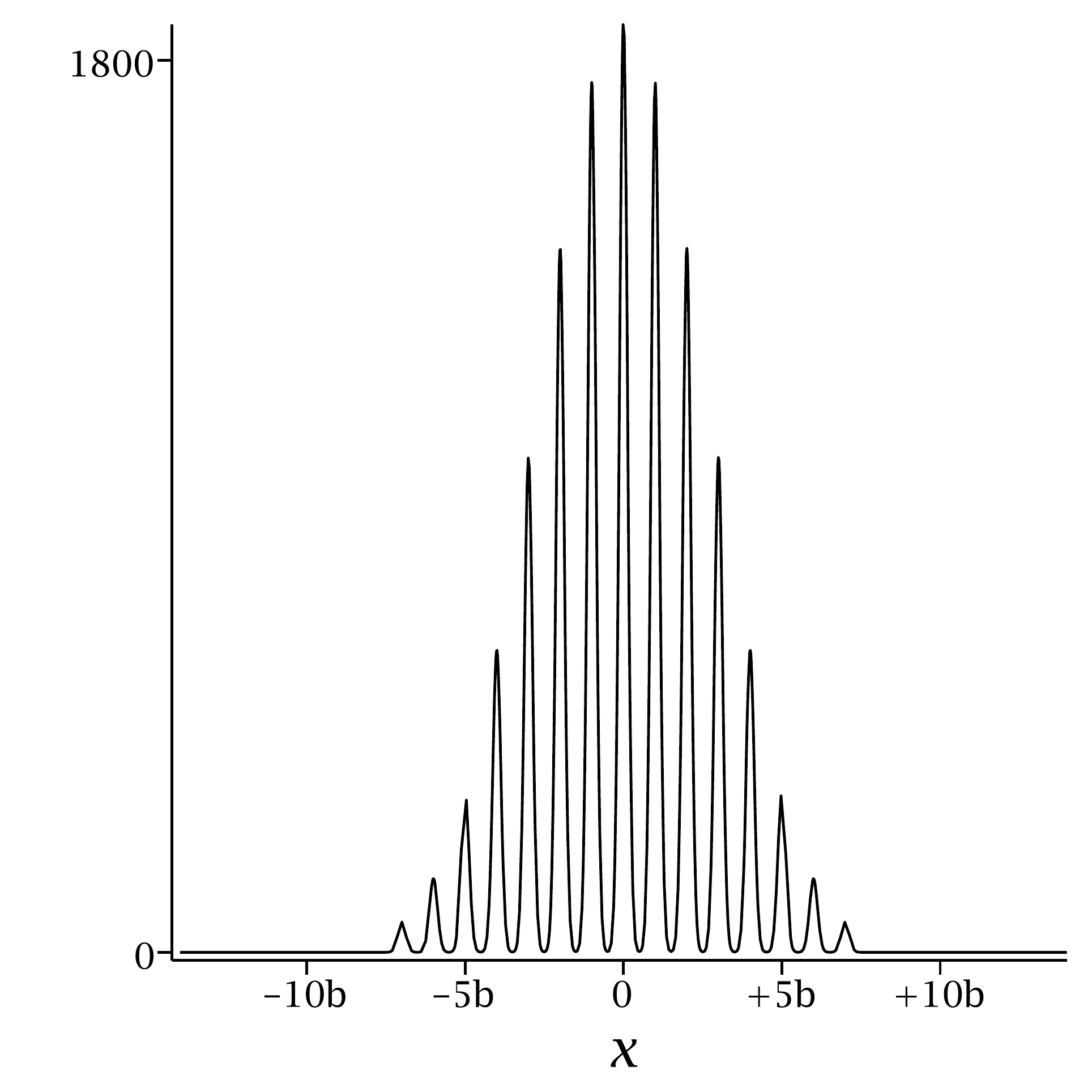}
\includegraphics[height=6cm,width=6cm]{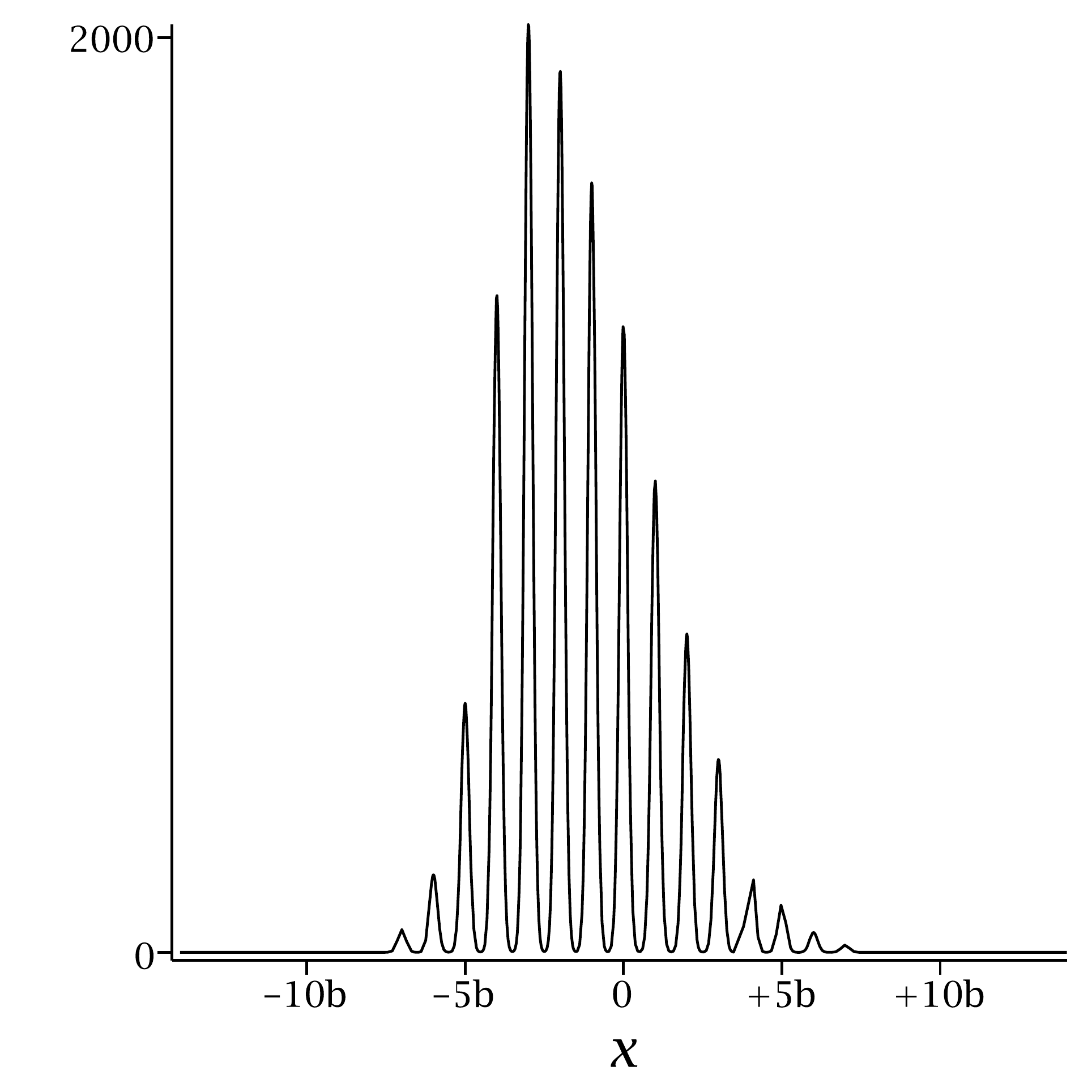}
\caption{\label {Figura5} Here we plot the absolute value of the initial wavefunction $\psi_0$ prepared on several lattice sites; the left hand side panel corresponds to 
the symmetric initial wavefunction, the right hand side panel corresponds to the asymmetrical one.}
\end{center}
\end{figure}
\begin{figure} [h]
\begin{center}
\includegraphics[height=6cm,width=12cm]{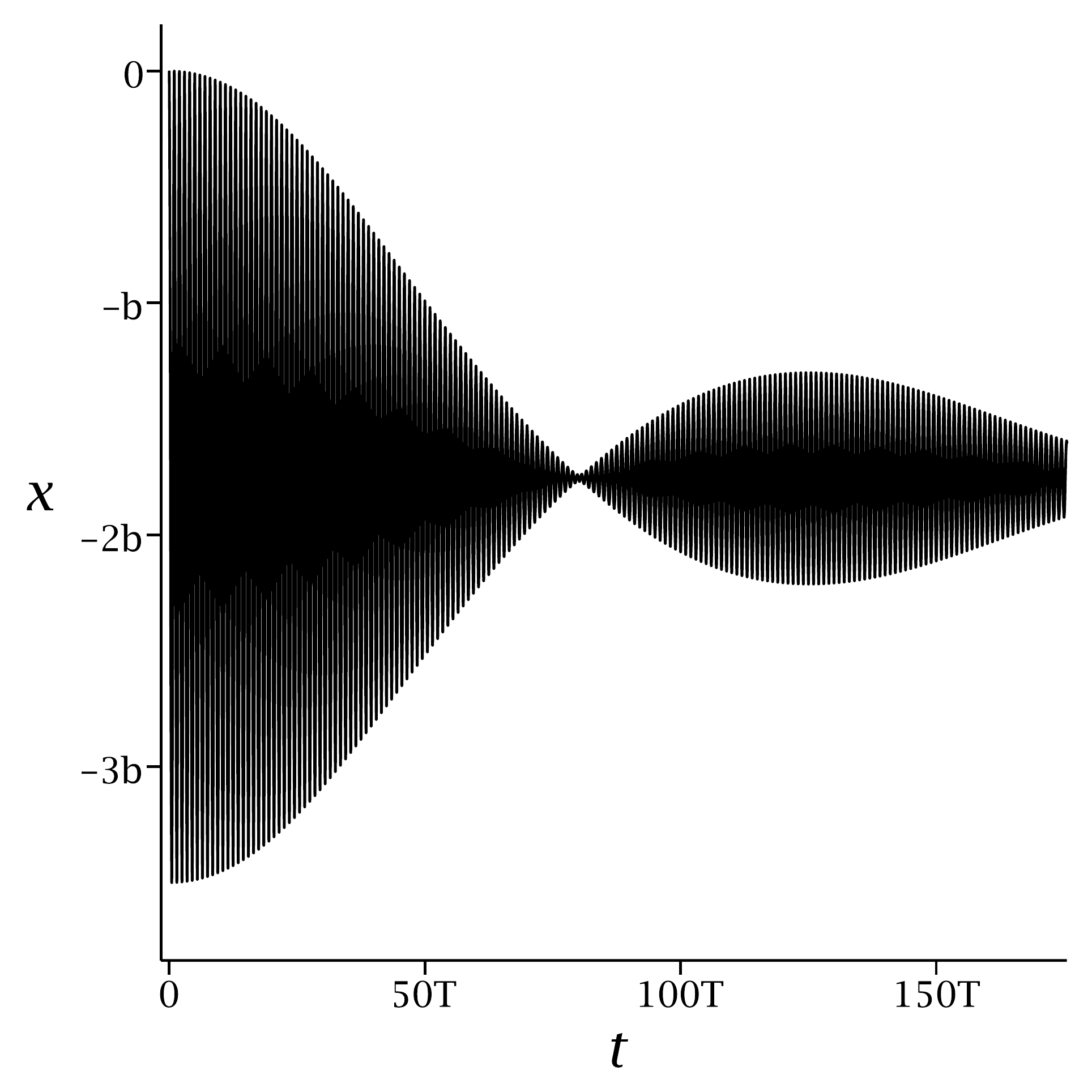}\\
\includegraphics[height=6cm,width=12cm]{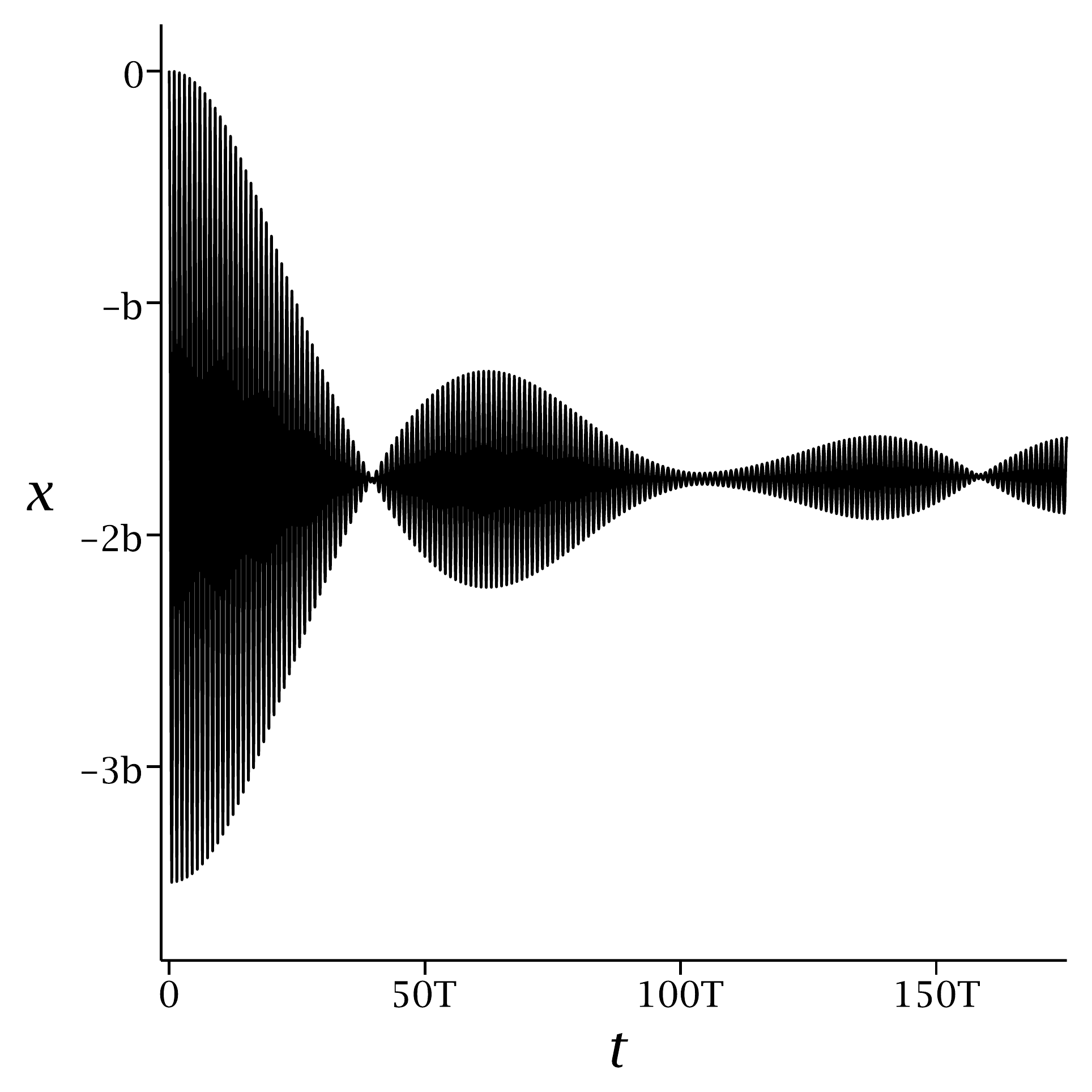}
\caption {\label {Figura6} Here we plot the motion of the center of mass of the wavefunction initially prepared on several lattice sites. \ The initial wavefunction 
is a symmetric function. \ Top panel corresponds to the case of $\eta =0.1$, bottom panel corresponds to the case of $\eta =0.2$. \ The center of mass rapidly oscillates 
with modulation of the amplitude. \ The width of the escillations is in a range lesser or equal to $3b$ is agreement with (\ref {BlochRange}).}
\end{center}
\end{figure}
\begin{figure} [h]
\begin{center}
\includegraphics[height=6cm,width=6cm]{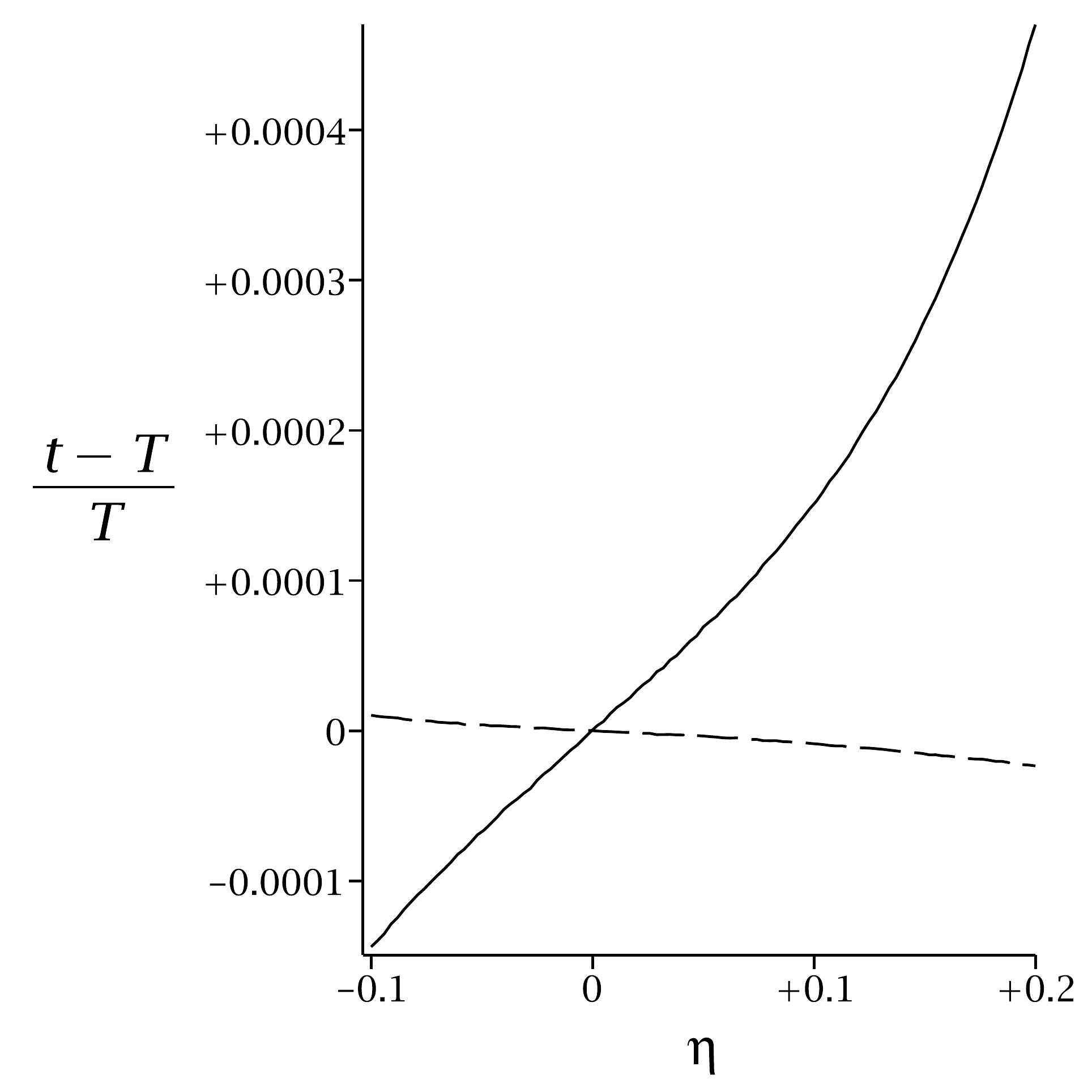}
\caption{\label {Figura7} Here we plot the mean value of the pseudo-period of the oscillating motion of the center of mass after 14 oscillations, as function of the effective nonlinearity parameter $\eta$. \ Broken line corresponds to the case of a 
symmetric wavefunction prepared on several lattice sites; it turns out that in such a case the oscillating period is almost constant. \ Full line  
corresponds to the case of an asymmetrical wavefunction prepared on several lattice sites; it turns out that it actually depends on $\eta$. \ Here 
$T$ denotes the Bloch period, while $t$ denotes the oscillating period.}
\end{center}
\end{figure}
\subsubsection {$\psi_0$ is an asymmetrical wavefunction initially prepared on different lattice cells} We consider a numerical experiment where $N=40$ and 
$\psi_0 (x) = \sum_{\ell=0}^{40} c_\ell u_\ell (x)$, where $c_\ell$ have an asymmetrical Gaussian-type distribution. \ That is the initial value of 
the coefficients $c_\ell (t)$ are given in Table \ref {tabella2}, the initial wavefunction is plotted in Fig. \ref {Figura5}, right hand side panel. \ As in the 
symmetric case 
the center of mass $\langle x \rangle^t$ oscillates in space and the wavefunction moves with no marked changes in shape. \ Even in such a case the 
function $\langle x \rangle^t$ exhibits, for $\eta \not= 0$, an oscillating motion where the wavefunction amplitude is modulated. \ In contrast with the symmetric case 
the oscillating (pseudo-)period (that is the time interval between two consecutive minima or maxima points) actually depends on $\eta$; 
in Fig. \ref {Figura7} we plot the mean value of the oscillating period of the center of mass after 14 oscillations for $\eta$ in the range $[-0.1,+0.2]$ 
and it is not almost constant like in the previous case, in particular it turns out that the relative uncertainty 
with respect to the Bloch period is of order $4.6 \cdot 10^{-4}$, which is 20 times the relative uncertainty observed in the symmetrical case.
\begin{table}
\begin{center}
\begin{tabular}{|c|c|c|c|} 
\hline
$c_0 =0$                     &$c_{10} =0.180 \cdot 10^{-3} $  &$c_{21}=0.252$                &$c_{31}=0.175 \cdot 10^{-4}$ \\ \hline
$c_1 =0$                     &$c_{11} =0.814 \cdot 10^{-3} $  &$c_{22}=0.170$                &$c_{32}=0.323 \cdot 10^{-5}$ \\ \hline
$c_2 =0$                     &$c_{12} =0.330 \cdot 10^{-2} $  &$c_{23}=0.103$                &$c_{33}=0$ \\ \hline
$c_3 =0$                     &$c_{13} =0.121 \cdot 10^{-1} $  &$c_{24}=0.546 \cdot 10^{-1}$  &$c_{34}=0$ \\ \hline
$c_4 =0$                     &$c_{14} =0.414 \cdot 10^{-1} $  &$c_{25}=0.257 \cdot 10^{-1}$  &$c_{35}=0$ \\ \hline
$c_5 =0$                     &$c_{15}= 0.133               $  &$c_{26}=0.106 \cdot 10^{-1}$  &$c_{36}=0$ \\ \hline
$c_6 =0$                     &$c_{16}= 0.351               $  &$c_{27}=0.386 \cdot 10^{-2}$  &$c_{37}=0$ \\ \hline
$c_7 =0$                     &$c_{17}= 0.496               $  &$c_{28}=0.123 \cdot 10^{-2}$  &$c_{38}=0$ \\ \hline
$c_8 =0.614 \cdot 10^{-5} $  &$c_{18}= 0.471               $  &$c_{29}=0.340 \cdot 10^{-3}$  &$c_{39}=0$ \\ \hline
$c_9 =0.354 \cdot 10^{-4} $  &$c_{19}= 0.411               $  &$c_{30}=0.826 \cdot 10^{-4}$  &$c_{40}=0$ \\ \hline 
  &$c_{20}=0.336$  &  & \\ \hline  
\end{tabular}
\caption{\small Initial values of the coefficients $c_\ell :=c_\ell (0)$ of the wavefunction. \ The initial wavefunction $\psi_0$ has an asymmetrical shape 
and its width is of order of several lattice periods.}
\label{tabella2}
\end{center}
\end {table}
\begin{figure} [h]
\begin{center}
\includegraphics[height=6cm,width=12cm]{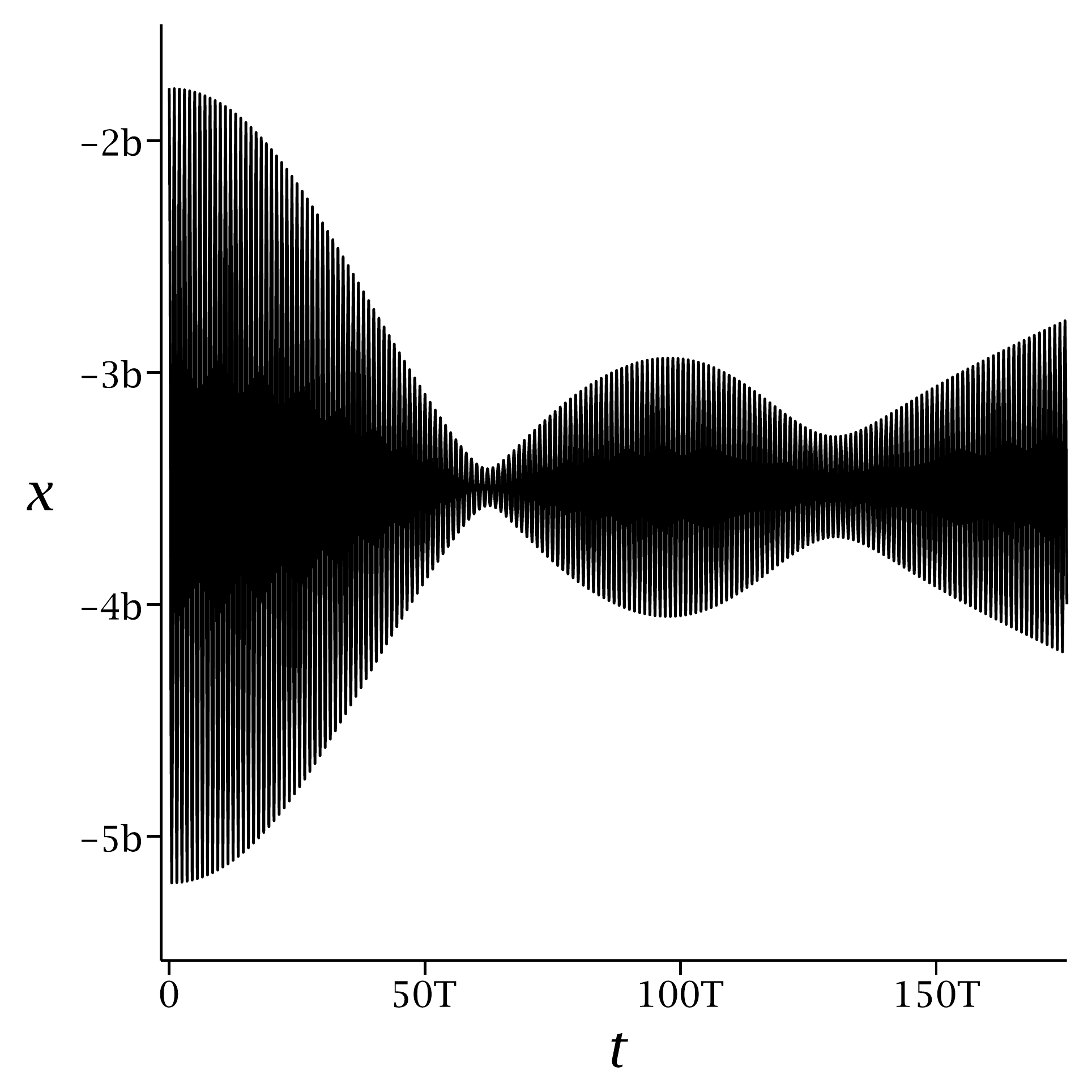} \\ 
\includegraphics[height=6cm,width=12cm]{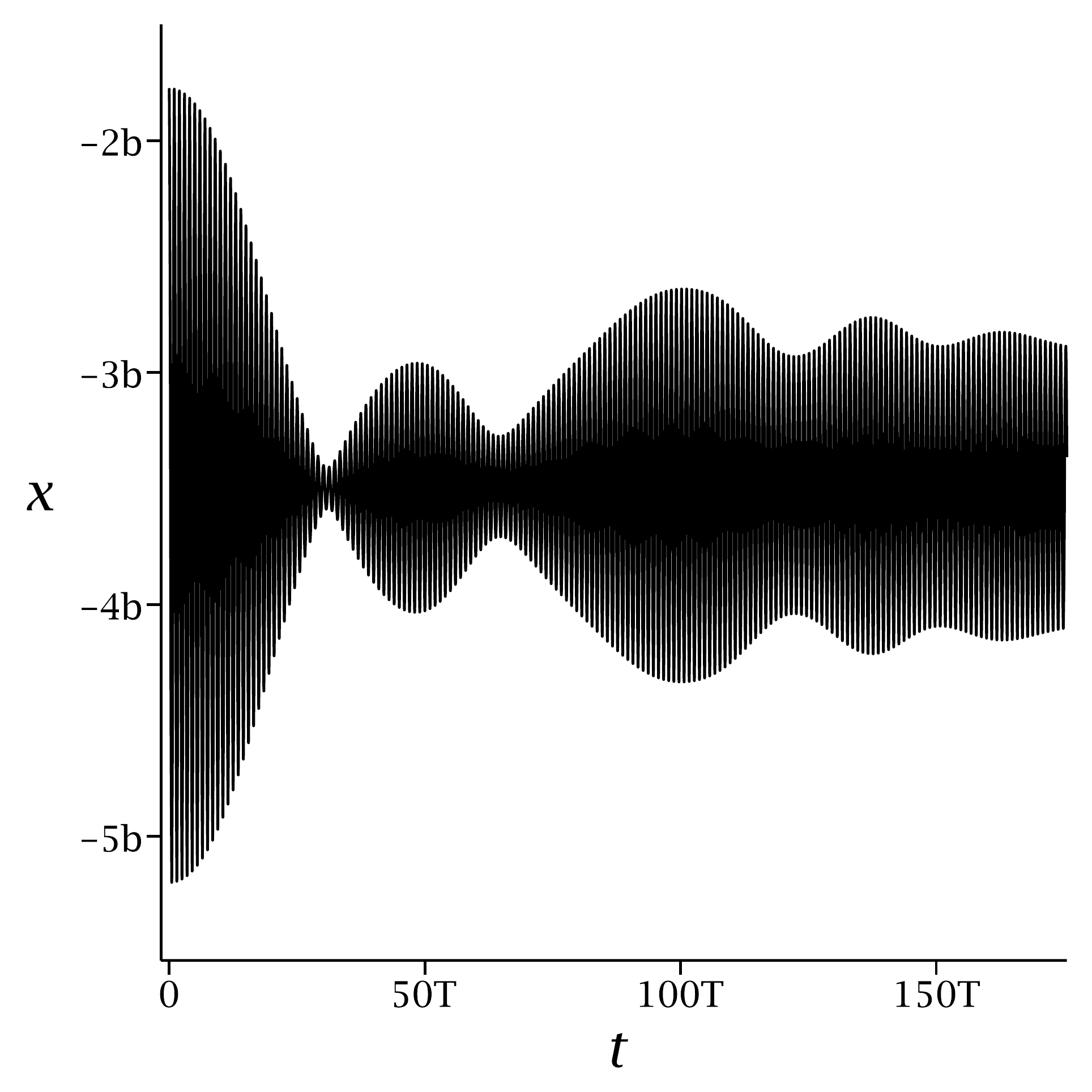}
\caption{\label {Figura8} Here we plot the motion of the center of mass of the wavefunction initially prepared on several lattice sites. \ The initial wavefunction 
is a symmetric function. \ Top 
panel corresponds to the case of $\eta =0.1$, bottom panel corresponds to the case of $\eta =0.2$. \ The center of mass rapidly oscillates 
with modulation of the amplitude. \ The width of oscillation is in a range lesser or equal to $3b$ is agreement with (\ref {BlochRange}).}
\end{center}
\end{figure}

\section {Conclusion} In this paper we have proved that in the semiclassical limit the $N$-mode approximation (\ref {Eq16}), corresponding to a discrete 
nonlinear Schr\"odinger equation with a finite number of modes, gives the solution of the Gross-Pitaevskii equation (\ref {Eq3}) for a BEC in a 
multiple-well lattice in a Stark-type external field. \ Furthermore, we have numerically solved the $N$-mode approximation considering a real model, where 
for some values of the physical parameters the validity of the $N$-mode approximation (\ref {Eq16}) seems to be justified. \ In particular, we have seen that a 
state initially prepared on several wells have an oscillating behavior with modulated amplitude, the oscillating (pseudo-)period is computed for different values 
of the nonlinear strength and it turns out that such a period is practically constant when the initial state is a symmetric one; on the other side, such a period 
actually depends on the nonlinear strength when the initial state is an asymmetrical one. \ This observation 
opens a question about the validity of the method proposed by Clad\'e {\it et al} \cite {CGSNJB} for the deterimantion of the gravitational constant $g$ by means of 
the measurement of the oscillating period \cite {FPST,PWTAPT}, where it has been assumed that the oscillating period coincides with the Bloch period $T$ independently from 
the shape of the initial wavefunction and of the value of the nonlinearity strength parameter. 

\appendix

\section {Band Functions and Wannier functions}

For a generic one-dimensional Bloch operator $H_B$ the spectrum is given by a sequence of infinitely many closed intervals named bands. \ These intervals are the image 
of functions named band functions. \ The band functions of $H_B$ are denoted by ${E}_n (k)$, where the quasimomentum $k$ runs in the Brillouin zone 
$\left [ - \frac {\pi}{b}, + \frac {\pi}{b} \right ]$. \ The spectrum of the Bloch operator $H_B$ is given by the bands $\sigma (H_B) = 
\cup_{n=1}^{\infty} [ {E}_n^b , {E}_n^t ]$ where 
\be
{E}_n^b = 
\left \{
\begin {array}{ll}
 {E}_n (0) & \ \mbox { if } n \mbox { is even} \\ 
 {E}_n (\pi/b) & \ \mbox { if } n \mbox { is odd}
\end {array}
\right.
\ , \ 
{E}_n^t = 
\left \{
\begin {array}{ll}
 {E}_n (\pi /b) & \ \mbox { if } n \mbox { is even} \\ 
 {E}_n (0) & \ \mbox { if } n \mbox { is odd}
\end {array}
\right. \, .
\ee

In the case of potential (\ref {Eq1}) the band functions may be explicitly computed. \ In particular let us look for the Bloch functions of the equation
\be
H_B \psi = {E} \psi \, ,  \ \ H_B =  - \frac {\hbar^2}{2m}\frac {d^2}{dx^2} + V_0 \sin^2 (k_L x) \, .
\ee
If we set 
\be
\E = \left ( {E} - \frac 12 V_0 \right ) \frac {2m}{\hbar^2}\, , \ q = 2 k_L \, , {\tilde V}_0= \frac {V_0 m}{\hbar^2} = 
\frac {m \Lambda_0 E_r }{\hbar^2} = \frac 12 \Lambda_0 k_L^2   
\ee
and recalling that $\sin^2 (\theta ) = \frac 12 \left [ 1- \cos (2\theta )\right ]$ then the Mathieu equation takes the form
\bee
\left [ \tilde H_B - \E \right ] \psi =0 \ \mbox { where } \ \tilde H_B = - \frac {d^2}{dx^2} - {\tilde V}_0 \cos ( q x) \, . \label {Eq24}
\eee
It has a fundamental set of solutions \cite {AS}
\be
\psi_1 (x,\E) = C \left [ \frac {4\E}{q^2} , - \frac {{\tilde V}_0}{q^2} , \frac 12 q x \right ] \ \mbox { and } \  
\psi_2 (x,\E) = \frac {2}{q}S \left [ \frac {4\E}{q^2} , - \frac {{\tilde V}_0}{q^2} , \frac 12 q x \right ] 
\ee
where $S$ and $C$ denotes the two Mathieu's functions, satisfying the conditions
\be
\psi_1 (0,\E) =1 \, , \frac {\partial \psi_1 (0,\E) }{\partial y} =0 \ \mbox { and } \ 
\psi_2 (0,\E) =0 \, , \frac {\partial \psi_2 (0,\E) }{\partial y} =1 \, . 
\ee

Hence, the band functions ${\mathcal E}_n (k)$ associated to the spectral problem (\ref {Eq24}) are the solutions 
of the equation $\mu (\E) = \cos (k b)$ where 
\be
\mu (\E) = \psi_1 (b,\E) = C \left [ \frac {4\E}{q^2} , - \frac {{\tilde V}_0}{q^2} , \pi \right ] \, . 
\ee

Let $\lambda =e^{i k b}$, then the equation $\mu (\E) = \cos (k b)$ can be written as $\mu (\E) = \frac 12 (\lambda + \lambda^{-1})$. \ We observe that 
for $k \in \left [ 0 , \frac {\pi}{b} \right ]$ then $\sin ( k b ) = \sqrt {1-\mu^2 (\E)}$. \ The Bloch function is given by \cite {K}
\be
\psi (x,\E) = \frac {\chi (x,\E)}{\sqrt {N (\E)}}
\ee
where
\be
\chi (x,\E) &=& \psi_2 (b,\E)\psi_1 (x,\E) + \frac 12 [ \lambda (\E) - \lambda^{-1} (\E) ] \psi_2 (x,\E) \\ 
&=& \psi_2 (b,\E)\psi_1 (x,\E) + i \sqrt {1-\mu(\E)^2} \psi_2 (x,\E)
\ee
and
\be
N (\E)= - \frac {4\pi}{b} \psi_2 (b,\E) \frac {d\mu }{d\E} \, . 
\ee
We recall that the Bloch function $\psi_n (x,k) = \psi (x, \E_n (k))$, where $\E_n$ is the band function associated to $\tilde H_B$, is normalized to one:
\be
\frac {2\pi}{b} \int_0^b |\psi_n (x,k) |^2 d x = 1 
\ee
and furthermore it is such that 
\be
\psi_n (x,-k ) = \overline {\psi_n (x,k)} \, . 
\ee

Finally, the Wannier function on the zero-th cell associated to the $n$-th band is given by 
\be
w_n (x) &=& \left ( \frac {b}{2\pi} \right )^{1/2} \int_{-\pi/b}^{+\pi/b} \psi_n (x,k) dk 
= 2 \left ( \frac {b}{2\pi} \right )^{1/2} \int_0^{+\pi/b} \Re \psi_n (x,k) dk \\ 
&=& 2 \left ( \frac {b}{2\pi} \right )^{1/2} \int_0^{\pi /b} \frac {\psi_2 (b,\E_n (k)) \psi_1 (x,\E_n (k))}{\sqrt {N(\E_n (k))}} dk \\ 
&=& \frac {b}{\sqrt {2}\pi} \int_0^{\pi /b} \frac {\sqrt {\psi_2 (b,\E_n (k)) }\psi_1 (x,\E_n (k))}{\sqrt {- \frac {d\mu (\E_n(k))}{d\E}}} dk \\ 
&=& \frac {1}{\sqrt {2}\pi} \int_{\E_n(0)}^{\E_n(\pi /b)} \frac {\sqrt {\psi_2 (b,\E) }\psi_1 (x,\E)\sqrt {- \frac {d\mu (\E)}{d\E}}}{\sqrt {1-\mu^2 (\E)}}d\E
\ee
since the Mathieu functions are real valued when their arguments are real numbers. \ In particular,
\be
u_0^{W}(x):= w_1 (x) = \frac {1}{\sqrt {2}\pi} \int_{\E_1^b}^{\E_1^t} \frac {\sqrt {\psi_2 (b,\E) }\psi_1 (x,\E)
\sqrt {- \frac {d\mu (\E)}{d\E}}}{\sqrt {1-\mu^2 (\E)}}d\E \, . 
\ee

{\bf Acknowledgements.} {\it This work is partially supported by Gruppo Nazione per la Fisica Matematica (GNFM-INdAM). \ The author is grateful for 
the hospitality of the Isaac Newton Institute for Mathematical Sciences where part of this paper was written.}

\end {document}